\documentclass[12pt]{article}
\usepackage{array}
\usepackage{bm}
\usepackage{amsmath}
\usepackage{graphicx,epsf}
\usepackage{enumerate}
\usepackage{natbib}
\usepackage{epstopdf}
\usepackage{epsfig}
\usepackage[figuresright]{rotating}
\usepackage{booktabs}
\usepackage{url} 
\usepackage{xcolor} 

\newcommand{\blind}{1}

\newcommand\cites[1]{\citeauthor{#1}'s\ (\citeyear{#1})}
\DeclareMathOperator*{\argmin}{arg\,min}

\def\T{{ \mathrm{\scriptscriptstyle T} }}

\def\QRCM{\textsc{qrcm}}
\def\ILM{\textsc{ilm}}

\def\thetavec{\bm{\theta}}
\def\phivec{\bm{\phi}}
\def\varphivec{\bm{\varphi}}
\def\betavec{\bm{\beta}}
\def\gammavec{\bm{\gamma}}
\def\aa{\bm{\alpha}_N}
\def\aahat{\hat{\bm{\alpha}}_N}

\def\xx{\bm{x}}
\def\zz{\bm{z}}
\def\bvec{\bm{b}}
\def\Bvec{\bm{B}}
\def\cvec{\bm{c}}
\def\Cvec{\bm{C}}
\def\vvec{\bm{v}}
\def\I{\bm{I}}

\def\G{\bm{G}}
\def\H{\bm{H}}
\def\O{\bm{\Omega}}

\def\csi{\bm{\xi}}

\def\Ep{\mathrm{E}}
\DeclareMathOperator{\plim}{plim}
\DeclareMathOperator{\ve}{vec}
\newtheorem{assumption}{Assumption}
\newtheorem{theorem}{Theorem}
\newenvironment{proof}[1][Proof]{\textbf{#1.} }{\ \rule{0.5em}{0.5em}}

\usepackage{amssymb}
\usepackage{amsfonts}

\begin{document}

\def\spacingset#1{\renewcommand{\baselinestretch}%
{#1}\small\normalsize} \spacingset{1}


\if1\blind
{
  \title{\bf PARAMETRIC MODELING\\ OF QUANTILE REGRESSION COEFFICIENT FUNCTIONS\\ WITH LONGITUDINAL DATA}
  \author{Paolo Frumento,  \vspace{.2cm}\\
    Matteo Bottai \vspace{0.2cm} \\
    and  \vspace{0.2cm} \\
    Iv\'an Fern\'andez-Val \vspace{0.5cm} \\
    University of Pisa, \\
    Karolinska Institutet\\
    and Boston University}
  \maketitle
} \fi

\if1\blind
{
  \newpage
  \title{\bf Author's Footnote} \vspace{1 cm}\\
  Paolo Frumento (paolo.frumento@unipi.it) is Associate Professor at the Department of Political Sciences, University of Pisa, Italy.\vspace{0.5 cm}\\
  Matteo Bottai (matteo.bottai@ki.se) is Professor at the Unit of Biostatistics at Karolinska Institutet, Institute of Environmental Medicine, Stockholm, Sweden.\vspace{0.5 cm}\\
  Iv\'an Fern\'andez-Val (ivanf@bu.edu) is Professor at the Department of Economics, Boston University.\vspace{0.5 cm}\\
  Address for correspondence: Department of Political Sciences, University of Pisa, Via F. Serafini, 3, 56126 Pisa, Italy.
  \newpage
} \fi

\if0\blind
{
  \bigskip
  \bigskip
  \bigskip
  \begin{center}
    {\LARGE \bf PARAMETRIC MODELING \\ OF QUANTILE REGRESSION COEFFICIENT FUNCTIONS \vspace{0.3cm} \\ WITH LONGITUDINAL DATA}
\end{center}
  \medskip
} \fi


\bigskip
\begin{abstract}
In ordinary quantile regression, quantiles of different order are estimated one at a time.
An alternative approach, which is referred to as \textit{quantile regression coefficients modeling} (\QRCM), 
is to model quantile regression coefficients as parametric functions of the order of the quantile.
In this paper, we describe how the $\QRCM$ paradigm can be applied to longitudinal data.
We introduce a two-level quantile function, in which two different quantile regression models
are used to describe the (conditional) distribution of the within-subject response and that
of the individual effects. We propose a novel type of penalized fixed-effects estimator,
and discuss its advantages over standard methods based on $\ell_1$ and $\ell_2$ penalization. We provide model identifiability conditions, derive asymptotic properties,
describe goodness-of-fit measures and model selection criteria, present simulation results, and 
discuss an application. The proposed method has been implemented in the R package \texttt{qrcm}.
\end{abstract}

\noindent%
{\it Keywords:} Longitudinal quantile regression, two-level quantile function, parametric quantile function, penalized fixed-effects, R package \texttt{qrcm}.
\vfill

\newpage
\spacingset{1.45} 


\section{INTRODUCTION}

Quantile regression (e.g., \citealp{kb, koenker}) has become a standard method 
in many fields, including medicine, epidemiology, economics, and social sciences. 
Different solutions have been proposed to extend quantile regression to 
longitudinal data, in which the same individuals or clusters are observed repeatedly.

In conditional models, that include fixed- and random-effects models, 
the dependence between observations is accounted for by introducing individual-specific parameters,
or ``individual effects''. In fixed-effects models, the individual effects are treated as parameters, avoiding
distributional assumptions and allowing for a simple computation. 
A penalized fixed-effects estimator for longitudinal quantile regression has been proposed by \cite{koenker2004}, 
and similar approaches have been used in \cite{lamarche}, \cite{canay}, and \cite{kato}.\footnote{\cite{cfw18} considered an alternative to quantile regression for estimation of quantile effects in longitudinal data based on  distribution regression.}
In random-effects models, the individual effects are described by a parametric distribution.
Different methods have been proposed to combine the parametric likelihood of the random effects with the estimating equation 
of ordinary quantile regression. \citeauthor{geraci} (\citeyear{geraci}, \citeyear{geraci2}) used the log-likelihood 
of an asymmetric Laplace distribution, and \cite{kim} described an empirical likelihood method. 
\cite{ad08} adapted the correlated random effect approach of \cite{c84} to quantile regression, and \cite{arellano} marginalized the loss function of quantile regression with respect to the posterior distribution 
of the individual effects. \cite{farcomeni}, \cite{marino}, and \cite{alfo} used finite mixtures to approximate the probability density function 
of the individual effects through a discrete distribution.

Marginal models have also been described in the literature.
\cite{leng} defined a set of unbiased estimating equations carrying information on the correlation structure.
A similar approach was used by \cite{zhao} to implement longitudinal single-index quantile regression.

In this paper, we adopt the conditional paradigm and introduce a two-level quantile function, in which
both the distribution of the within-subject response (level 1) and that
of the individual effects (level 2) are described by quantile regression models. 
With this approach, the distribution of the individual effects is not subject to strong parametric 
assumptions and is allowed to depend on level-2 covariates.
Following \citeauthor{iqr} (2016, 2017) and \cite{yang}, we describe the level-1 and level-2 quantile
regression coefficients by (flexible) parametric functions of the order of the quantile.
Compared with standard quantile regression, in which quantiles are estimated one at a time, this modeling approach 
presents numerous advantages, that include a simpler computation and inference
(owing to a smooth objective function), increased statistical efficiency, and easy interpretability of the results.

To fit the model, we introduce a new form of penalized fixed-effects estimator in which the penalty term
carries information on level-2 parameters. This method presents important advantages 
over standard $\ell_1$ and $\ell_2$ penalization. In particular, it avoids the problem of selecting a tuning constant,
and allows to estimate the level-2 coefficients of the model using fixed-effects techniques.

The paper is structured as follows. We describe a general model in Section \ref{sec:model0},
and discuss model building in Section \ref{sec:model}. We introduce an
estimator in Section \ref{sec:est}, and in Section \ref{sec:asy} we derive its asymptotic properties. 
In Section \ref{sec:gof} we present goodness-of-fit measures and tools for model selection,
and in Section \ref{sec:sim} we report simulation results. Section \ref{sec:app} concludes the paper with the analysis of 
a dataset relating plasma neutrophil gelatinase-associated lipocalin (NGAL) to sepsis status. Appendix A provides a general asymptotic expansion for fixed effects estimators with mixed-rates asymptotics and applies it to derive the asymptotic distribution of the proposed estimator.
We discuss computation in Appendix B, 
and present extended simulation results in Appendix C. The R package \texttt{qrcm} implements the described estimator
and provides a variety of auxiliary functions for model building, summary, plotting, prediction, and goodness-of-fit assessment.


\section{THE MODEL}\label{sec:model0}

\subsection{A two-level quantile function}
We consider a cluster data structure, in which $N$ individuals or clusters are observed repeatedly.
We denote by $i = 1, \ldots, N$ the index of the subject, and by $t = 1, \ldots, T$ the within-subject index,
such that the total sample size is $NT$. Designs in which $T$ varies across clusters are also possible,
at the cost of a slightly more complicated notation.

We denote by $Y_{it}$ a response variable of interest, and assume that
\begin{equation}\label{themodel}
Y_{it} = \xx_{it}^\T\betavec(U_{it}) + \zz_{i}^\T\gammavec(V_{i})
\end{equation}
where $\xx_{it}$ is a $d_x$-dimensional vector of level-1 covariates, with associated parameter $\betavec(\cdot)$;
and $\zz_i$ is a $d_z$-dimensional vector of level-2 covariates, with associated parameter $\gammavec(\cdot)$.

We assume that (i) $\xx_{it}^\T\betavec(\cdot)$ and $\zz_{i}^\T\gammavec(\cdot)$ are a.s. non-decreasing functions of their arguments, and
(ii) $U_{it}$ and $V_i$ are $U(0,1)$ variables, independent of each other and of the covariates.
Based on model (\ref{themodel}), $\alpha_i = \zz_{i}^\T\gammavec(V_{i})$ is an individual effect with 
conditional quantile function $\zz_{i}^\T\gammavec(\cdot)$, while $\xx_{it}^\T\betavec(\cdot)$ is the conditional 
quantile function of $Y_{it} - \alpha_i$. 

The level-1 quantile regression model, $\xx_{it}^\T\betavec(\cdot)$, has the standard interpretation 
(e.g., \citealp{koenker2004}): it characterizes the ``within'' part of the 
distribution, purged of the individual effects.
The level-2 regression model, $\zz_{i}^\T\gammavec(\cdot)$,
describes the distribution of the between-subject differences with respect to a reference value
which typically corresponds to a ``mean'' or ``median'' individual.

Consider, for example, a clinical study in which patients are 
repeatedly measured their body mass index (BMI) during their lifetime.
The level-1 part of the model describes the conditional quantiles of BMI in a ``typical''
patient, i.e., someone with an individual effect equal to $0$.
Level-1 predictors include time-varying characteristics, such as the age of the patient at each observation,
as well as constant traits, such as the gender of the patient.
The level-2 model accounts for the between-patient heterogeneity, and describes the 
conditional quantiles of the individual effects. Level-2 covariates 
can only include time-invariant traits, such as the gender,
and summary statistics of level-1 covariates, e.g., the age at the first examination.
Note that the dimension of the level-1 covariates, $\xx_{it}$,
is $NT$, while that of the level-2 covariates, $\zz_i$, is $N$.

Unlike the ``standard'' approaches, that do not consider the effect of level-2 covariates, 
our modeling framework allows to investigate the determinants of the between-subject variability.
For example, in the linear random-intercept model, the level-2 response is described by a $N(0, \phi^2)$ distribution,
in which $\phi^2 = \text{var}(\alpha_i)$ is interpreted as the ``between'' variance and 
is assumed to be unaffected by predictors. This model may fail to capture important features of the data,
such as the fact that the variance of the individual effects is different in males and females. 
Model \eqref{themodel}, instead, allows including gender as level-2 predictor.

Using a quantile regression approach permits avoiding strong parametric assumptions such as normality and homoskedasticity,
that are often used in likelihood-based modeling. In the existing literature on longitudinal quantile regression, however,  
a quantile regression model is usually applied to the level-1 response, but not to the individual effects, that are treated as nuisance parameters.
In our paradigm, instead, the two parts of the distribution are considered ``equally important'',
in the sense that the same modeling structure is used to describe 
the quantiles of the within-subject response, and those of the between-subject differences.
As shown later in the paper, working with model \eqref{themodel} permits using the same techniques to estimate both the level-1 and the level-2 parameters,
and avoids combining level-1 quantile regression methods with likelihood-based level-2 estimators as for example in \cite{kim}.
This leads to rather simple procedures for estimation and inference, in which a fundamental role is played by the two independent uniform random variables 
($U_{it}, V_i$) that generate the data.

\subsection{Parametric coefficient functions}
Through the paper, we assume that
the quantile regression coefficient functions, $\betavec(\cdot)$ and $\gammavec(\cdot)$,
can be modeled parametrically:
\begin{equation}\label{par}
  \betavec(u) = \betavec(u\mid\thetavec), \hspace{0.2cm}
  \gammavec(v) = \gammavec(v\mid\phivec),
\end{equation}
where $\thetavec$ and $\phivec$ are unknown  model parameters. 
This modeling approach was used by \citeauthor{iqr} (2016, 2017),
and is exemplified in Figure~1. The broken line in figure 
represents standard regression coefficients at quantiles $u = (0.01,0.02, \ldots, 0.99)$. 
The estimated coefficients show a non-smooth, volatile trend and, 
although consistently positive, are almost never significant. A parametric model 
can be used to characterize the coefficient function with few parameters and describe it by a simple,
closed-form mathematical expression. In Figure~1 we propose a linear fit, 
$\beta(u \mid \thetavec) = \theta_0 + \theta_1 u$, that is represented by a dashed line.
This simple model reveals the underlying trend and permits achieving statistical significance.

Compared with standard quantile regression, which works in a quantile-by-quantile fashion,
modeling quantile functions parametrically simplifies estimation and inference and yields important 
advantages in terms of parsimony, efficiency, and ease of interpretation. Moreover,
it allows for model identification in the presence of latent structures or missing information,
making it simple to apply quantile regression to censored and truncated data \citep{ctiqr}.

On the other hand, this approach requires formulating a parametric model for the coefficient functions,
$\betavec(u\mid\thetavec)$ and $\gammavec(v\mid\phivec)$. This task is not straightforward and
the existing literature on the subject is lacking. In Section \ref{sec:model} we describe in details
model building, provide guidelines, and suggest a variety of possible parametrizations.

\begin{figure}[h]
    \centering    
	\makebox{\includegraphics[scale = 0.5]{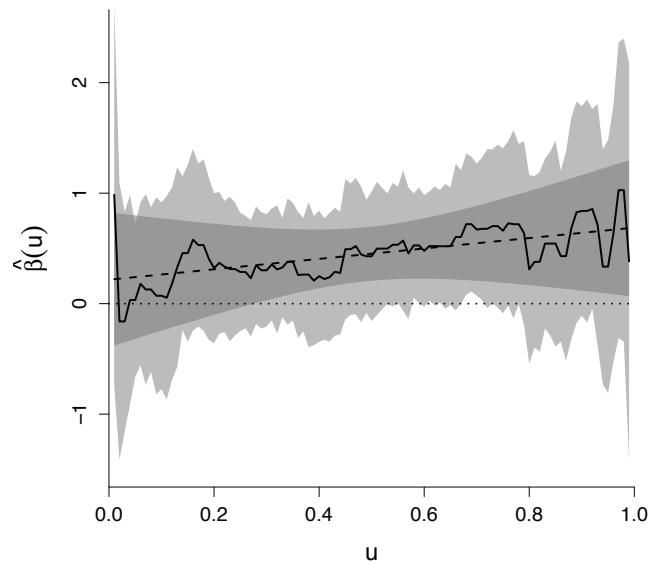}}
	\caption{\label{fig1}
	The broken line exemplifies a typical behavior of quantile regression estimators,
	while the dashed line suggests a hypothetical linear trend, $\beta(u \mid \thetavec) = \theta_0 + \theta_1 u$.
	Pointwise confidence intervals are represented by the light and dark shaded areas, respectively. 
	The dotted line indicates the zero.}
\end{figure}



\section{TWO-LEVEL MODELING OF QUANTILE REGRESSION COEFFICIENT FUNCTIONS}\label{sec:model}

We assume model (\ref{themodel}) to hold, and parametrize the quantile regression coefficient functions
as follows:
\begin{equation}\label{par2}
\betavec(u\mid\thetavec) = \thetavec\bvec(u), \hspace{0.2cm} \gammavec(v\mid\phivec) = \phivec\cvec(v),
\end{equation}
where $\bvec(u) = \left[b_1(u), \ldots, b_{d_{b}}(u)\right]^{\T}$ and $\cvec(v) = \left[c_1(v), \ldots, c_{d_{c}}(v)\right]^{\T}$
are $d_{b}$- and $d_{c}$-dimensional sets of known functions. With this notation, $\thetavec$ is a $d_{x}\times d_{b}$ matrix,
and $\phivec$ is a $d_{z}\times d_{c}$ matrix. The data-generating process can be written as
\begin{equation}\label{themodel2}
Y_{it} = \xx_{it}^\T\thetavec\bvec(U_{it}) + \zz_{i}^\T\phivec\cvec(V_i).
\end{equation}
Although other parametrizations are possible (e.g., $\betavec(u\mid\thetavec)$ and $\gammavec(v\mid\phivec)$
may be allowed to be nonlinear functions of $\thetavec$ and $\phivec$),
model \eqref{par2} is very flexible and computationally convenient. We illustrate the potentials
of this modeling approach with a number of examples, and provide general guidelines for model building.

\subsection{A simple model}\label{subsec:simplemodel}

Consider the following model with a single level-1 covariate $x$, and no level-2 predictors:
$$Y_{it} = \beta_0(U_{it}) + \beta_1(U_{it})x_{it} + \gamma_0(V_i).$$
Denote by $\zeta(\cdot)$ the quantile function of a standard normal distribution, and assume that
$$\beta_0(u\mid\thetavec) = \theta_{00} + \theta_{01}\zeta(u),$$
$$\beta_1(u\mid\thetavec) = \theta_{10},$$
$$\gamma(v \mid\phivec) = \phi \zeta(v).$$
This is just a reformulation of the standard linear random-intercept model, in which
$Y_{it} = \theta_{00} + \theta_{10}x_{it} + \alpha_i + \epsilon_{it}$ with $\alpha_i \sim N(0,\phi^2)$
and $\epsilon_{it} \sim N(0,\theta_{01}^2)$. In this model, 
$\theta_{00}$ corresponds to the intercept of the ``fixed'' part, while $\theta_{01}^2$ and $\phi^2$ are interpreted 
as the ``within'' and ``between'' variance components. In the equivalent quantile regression model, $\theta_{00}$ is the ``intercept'' of $\beta_0(u\mid\thetavec)$
and corresponds to $\beta_0(0.5\mid\thetavec)$, while $\theta_{01}$ and $\phi$ are ``slopes'' associated with $\zeta(\cdot)$ in the level-1 and level-2 
part of the quantile function, respectively. The regression coefficient of $x$, 
$\beta_1(u\mid\thetavec)$, is assumed to be constant across quantiles, forcing homoskedasticity.

\subsection{A more flexible model}\label{subsec:flexmodel}
The standard linear random-intercept model is rather restrictive and, within the described framework,
can be easily generalized by choosing a different specification of $\betavec(\cdot \mid \thetavec)$ and $\gammavec(\cdot \mid \phivec)$. 
For example, one may define
$$\beta_0(u\mid\thetavec) = \theta_{00} +  \theta_{01}u + \theta_{02}u^2 + \theta_{03}u^3 + \theta_{04}\zeta(u),$$
$$\beta_1(u\mid\thetavec) = \theta_{10} + \theta_{11}u,$$
$$\gamma(v \mid\phivec) = \phi_1 \log{(2v)} + \phi_2\log{(2(1 - v))}.$$
The intercept, $\beta_0(u\mid\thetavec)$, is modeled by a linear combination of $\zeta(u)$, the quantile function of a standard normal
distribution, and three additional components, $u$, $u^2$ and $u^3$, that allow for a deviation from the normal model. The resulting quantile function
can be asymmetric or multimodal and does not correspond to any ``standard'' family of random variables. 
The coefficient associated with $x$, $\beta_1(u\mid\thetavec)$, is now assumed to be a linear function of $u$,
allowing for data heteroskedasticity. In particular, the variance of the level-1 response is an increasing function of $x$,
if $\theta_{11} > 0$, and a decreasing function of it, if $\theta_{11} < 0$. Finally, the individual effects are assumed to follow
a zero-median asymmetric logistic distribution, which is much more flexible than the commonly used normal model.

As shown in this example, $\betavec(\cdot)$ and $\gammavec(\cdot)$ can be constructed as linear combinations of relatively
simple functions, $\bvec(\cdot)$ and $\cvec(\cdot)$, such that $\betavec(u\mid\thetavec) = \thetavec\bvec(u)$, and $\gammavec(v\mid\phivec) = \phivec\cvec(v)$.
In this framework, the model is entirely determined by the choice of $\bvec(\cdot)$ and $\cvec(\cdot)$.
Useful guidelines for model building are provided in the rest of this section.
Various modeling approaches are illustrated in Sections \ref{sec:sim} and \ref{sec:app} of this paper,
while a general discussion on quantile modeling can be found in the book by \cite{gil}.
Finally, the documentation of the \texttt{qrcm} package (in particular the functions
\texttt{iqr} and \texttt{iqrL}) includes an extensive tutorial for the practitioners.

\subsection{Model building: level 1}\label{subsec:level1}

\textbf{\textsc{Modeling $\beta_0(u \mid \thetavec)$}}. 
Assuming that the support of $\xx$ includes the zero (which can be obtained by centering the covariates), 
$\beta_0(\cdot \mid \thetavec)$ must be a monotonically increasing function.
Prior belief or knowledge can be used to identify a meaningful parametric model. 
For instance, one may use the quantile function of a known distribution. Possible parametrizations of $\beta_0(u \mid \thetavec)$ include:
$\theta_{00} + \theta_{01}\zeta(u)$, the normal distribution, $\text{N}(\theta_{00}, \theta_{01}^2)$;
$-\theta_{01}\log(1 - u)$, the exponential distribution, $\text{Exp}(\theta_{01})$;
$\theta_{00} + \theta_{01}\log(u/(1 - u))$, the logistic distribution, $\text{Logis}(\theta_{00}, \theta_{01})$;
$\theta_{00} + \theta_{01}\log(u) + \theta_{02}\log(1 - u)$, the asymmetric logistic, $\text{aLogis}(\theta_{00}, \theta_{01}, \theta_{02})$;
$\theta_{00} + \theta_{01}u$, the uniform distribution, $\text{U}(\theta_{00}, \theta_{00} + \theta_{01})$.
Note that, in this framework, the parameters of well-known distributions may have
an unusual interpretation. For example, the value of $\theta_{01}$ in a $\text{U}(\theta_{00}, \theta_{00} + \theta_{01})$
distribution corresponds to its range, but can also be seen as the slope of a linear quantile function, $\theta_{00} + \theta_{01}u$.\vspace{0.3cm}\\ 
\textbf{\textsc{Modeling $\beta_1(u \mid \thetavec), \beta_2(u \mid \thetavec), \ldots$.}} 
There are no general constraints to the parametric form of the regression coefficients associated with the covariates.
However, the coefficient functions are usually bounded and exhibit a rather simple behavior.
Sometimes, it is possible to assume that covariates only affect the location of the level-1 response, and force homoskedasticity by choosing
a constant-slope model in which $\beta_j(u \mid \thetavec) = \theta_{j0},  j = 1, 2, \ldots$. 
In a more general scenario, a useful approximation is often given by a linear-slope model, 
$\beta_j(u \mid \thetavec) = \theta_{j0} + \theta_{j1}u$, or a quadratic-slope model, $\beta_j(u \mid \thetavec) = \theta_{j0} + \theta_{j1}u + \theta_{j2}u^2$, which does not impose monotone effect with respect to $u$.\vspace{0.3cm}\\ 

\subsection{Model building: level 2}\label{subsec:level2}

A similar model strategy can be applied to the level-2 quantile function.
There are, however, some important differences.\vspace{0.3cm}\\ 
\textbf{\textsc{Modeling $\gamma_0(v \mid \phivec)$}}.
The distribution of the individual effects is typically assumed to have zero mean or median, and,
for identifiability, $\gamma_0(v \mid \phivec)$ does not usually include a constant term. Meaningful definitions
of $\gamma_0(v \mid \phivec)$ include:
$\phi_{01}\zeta(v)$, the normal distribution, $\text{N}(0, \phi_{01}^2)$;
$-\phi_{01}\log(1 - v)$, the exponential distribution, $\text{Exp}(\phi_{01})$;
$\phi_{01}\log(v/(1 - v))$, the logistic distribution, $\text{Logis}(0, \phi_{01})$;
$\phi_{01}\log(2v) + \phi_{02}\log(2(1 - v))$, a zero-median asymmetric logistic;
$\phi_{01}[\log(v) + 1] + \phi_{02}[\log(1 - v) + 1]$, a zero-mean asymmetric logistic;
$\phi_{01}2(v - 0.5)$, a centered uniform distribution, $\text{U}(-\phi_{01}, \phi_{01})$.
In most cases, the coefficients can be interpreted as scale parameters, while the centrality parameter is 
fixed and equal to zero. In the exponential case, the value $0$ is the minimum of the support of the individual effects, 
and not a measure of central tendency, while both 
the mean and the standard deviation of the individual effects correspond to $\phi_{01}$.\vspace{0.3cm}\\ 
\textbf{\textsc{Modeling $\gamma_1(v \mid \phivec), \gamma_2(v \mid \phivec), \ldots$.}} 
Importantly, the described framework permits investigating how the conditional quantile function of the individual effects
depends on level-2 covariates $\zz_i$, which typically include cluster-invariant characteristics (e.g., gender) or 
summary measures of the level-1 covariates, e.g., the cluster means or medians. 
The variance of the individual effects is likely to differ across subgroups of the population.
Also, as suggested by some authors (e.g., \citealp{lancaster}), agents may select their covariates' values based on prior knowledge about their own individual effect, 
which induces a correlation between $\alpha_i$ and $\zz_i$.

Modeling the effect of level-2 covariates is not trivial. To make an example, suppose that 
$\alpha_i = \gamma_0(V_i) + \gamma_1(V_i)z_i$, and consider the following alternative parametrizations:
\begin{equation}
\gamma_0(v \mid \phivec) = \phi_{01}\zeta(v), \hspace{0.2cm} \gamma_1(v \mid \phivec) = \phi_{11}\zeta(v)\tag{i}
\end{equation}
\begin{equation}
\gamma_0(v \mid \phivec) = \phi_{01}\zeta(v), \hspace{0.2cm} \gamma_1(v \mid \phivec) = \phi_{10} + \phi_{11}v \tag{ii}\\
\end{equation}
\begin{equation}
\gamma_0(v \mid \phivec) = \phi_{01}v, \hspace{0.2cm} \gamma_1(v \mid \phivec) = \phi_{10} + \phi_{11}v \tag{iii}\\
\end{equation}
In model (i), where $\gamma_0(v \mid \phivec)$ and $\gamma_1(v \mid \phivec)$ are symmetric around the zero,
the conditional distribution of $\alpha_i$ has zero mean and median at all values of $z_i$. The covariate only affects
the scale of the individual effects by introducing heteroskedasticity, while no linear correlation between $z_i$ and $\alpha_i$ is present. 
Model (i) assumes normality, but allows the variance of the individual effects to be a function of the level-2 covariates, i.e. $\alpha_i \mid z_i \sim \text{N}(0, \phi_{01}^2 + \phi_{11}^2 z_i^2)$.
For example, if $z_i$ is binary, the ``between'' variance is $\phi_{01}^2$ when $z_i = 0$, and $\phi_{01}^2 + \phi_{11}^2$ when $z_i = 1$.

In models (ii) and (iii), $z_i$ and $\alpha_i$ have a non-zero correlation unless $\phi_{10} = 0$.
In model (ii), where $\int_0^1\gamma_0(v \mid \phivec) \mathrm{d}v = 0$, the \textit{marginal} distribution of the individual 
effects has zero mean if $z_i$ is centered around its mean or $\phi_{10} + \phi_{11}/2 = 0$. In model (iii) the mean and the median
of the individual effects are functions of the parameters and cannot be determined in advance.  
However, if $z_i \ge 0$, model (iii) generates $\alpha_i \ge 0$ for any positive value of the parameters,
implying that the ``reference'' individual ($\alpha_i = 0$) corresponds to 
someone with the smallest possible individual effect.

\subsection{Additional remarks}\label{subsec:add}

The problem of formulating a parametric quantile function is equivalent, at least in principle, to that of choosing a parametric form for a probability density function, a hazard function, or a survival function. For example, as shown in Section \ref{subsec:simplemodel}, standard parametric assumptions such as normality and homoskedasticity can be directly translated into a quantile function with a simple closed-form expression. However, as suggested in Section \ref{subsec:flexmodel}, the models that can be used to describe a quantile function are often very different from most of the ``conventional'' parametric distributions, and frequently much more flexible.

An exploratory semiparametric fit can be
obtained by letting $\bvec(\cdot)$ and $\cvec(\cdot)$ be the basis of a linear or polynomial spline. A flexible model can be used
as a guide to find more parsimonious and efficient parametrizations. Note that standard quantile regression, in which 
quantiles are estimated one at a time, can be thought of as a model in which $\bvec(\cdot)$ and $\cvec(\cdot)$ are allowed to be arbitrarily flexible
and the parameters $\thetavec$ and $\phivec$ are virtually infinite-dimensional.

In absence of prior knowledge, one may define $\bvec(\cdot)$ and $\cvec(\cdot)$ using polynomials $\left[\text{e.g.,} u, u^2, u^3, \ldots\right]$, roots $\big[\text{e.g.,} u^{1/2}$, $(1 - u)^{1/2}$, $u^{1/3}$,  $(1 - u)^{1/3}, \ldots\big]$, trigonometric functions $\left[\text{e.g.,} \cos(2\pi u), \sin(2\pi u)\right]$, splines, and combinations of the above. A possible strategy is to consider a ``simple'' quantile function (e.g., that of a normal or an exponential distribution, depending on the nature of the outcome) and allow for a departure from it, as suggested in Section \ref{subsec:flexmodel}.

Importantly, the model specification should reflect assumptions on the shape, support, and boundedness (or unboundedness) of the level-1 and level-2 responses. For example, if the individual effects are believed to be symmetric, $\gamma_0(v \mid \phivec)$ could be formed by the quantile function of a normal or logistic distribution. If the level-1 distribution has a long right tail, $\beta_0(u \mid \thetavec)$ may have a positive asymptote in $u = 1$, e.g., $\beta_0(u \mid \thetavec) = \theta_{00} - \theta_{01} \log(1 - u) + \ldots$. On the other hand, if the outcome is strictly positive, building blocks such as $\log(u)$ or $\zeta(u)$, that present a negative asymptote in $u = 0$, may not be appropriate.

Apart from the above important considerations, the choice of $\bvec(\cdot)$
and $\cvec(\cdot)$ is not as crucial as it appears. For example, 
the coefficient function defined by $\beta(u) = (u - 0.3)^3$ is almost identical to $\beta(u) = -1.87 + 6.20u + 1.84\cos(u) - 5.92\sin(u)$,
the correlation between the two being about $0.99999$. The fact that very different model specifications can be indistinguishable in terms of model fit
is unsurprising (for example, it is almost impossible to distinguish a Normal distribution, a Student's t distribution with large degrees of freedoms, 
and a Gamma distribution with large shape parameter), and suggests that meaningful criteria for model selection should include parsimony 
and interpretability.

Often, a rather restrictive model may provide a reasonable approximation of the true data distribution, 
and can be preferred to a more correct model because of its simplicity. Also, parsimonious models are very
rewarding in terms of precision, although they may introduce some bias. This explains why
strong parametric assumptions, such as homoskedasticity and proportionality of hazards or odds, 
are used routinely in statistical analysis. In quantile regression, very convenient assumptions are represented by the constant-slope 
model (e.g., $\beta(u \mid \thetavec) = \theta_0$), in which a certain predictor has the same effect at all quantiles, and the
linear-slope model (e.g., $\beta(u \mid \thetavec) = \theta_0 + \theta_1u$), in which a quantile regression coefficient is assumed to be a linear function.



\section{THE ESTIMATOR}\label{sec:est}

\cite{iqr} considered cross-sectional data $(y_i, \xx_i)$ and defined $\betavec(u\mid\thetavec) = \thetavec\bvec(u)$
as in (\ref{par2}). To estimate $\thetavec$, they suggested minimizing
\begin{equation}\label{iloss}
	L(\thetavec) = \int_0^1{\sum_i{\rho_u(y_i - \xx_i^\T\betavec(u \mid \thetavec))} \mathrm{d}u},
\end{equation}
which is the integral, with respect to the order of the quantile, of the loss function of standard
quantile regression, $\rho_u(w) = w(u - I(w \le 0))$ being the ``check'' function. 
This estimation method is referred to as \textit{integrated loss minimization} (\ILM) 
and is currently implemented in the \texttt{qrcm} R package. 

To generalize this idea to longitudinal data,
assume model (\ref{themodel2}) holds,
\begin{equation}\nonumber
Y_{it} = \xx_{it}^\T\betavec(U_{it} \mid \thetavec) + \zz_{i}^\T\gammavec(V_{i} \mid \phivec) = \xx_{it}^\T\thetavec\bvec(U_{it}) + \zz_{i}^\T\phivec\cvec(V_i),
\end{equation}
and denote by $y_{it}$ a realization of $Y_{it}$.
If the individual effects $\alpha_i = \zz_{i}^\T\phivec\cvec(V_i)$ were known, one could directly apply the $\ILM$ estimator to $y_{it} - \alpha_i$,
to compute an estimate of $\thetavec$; and to $\alpha_i$, to compute an estimate of $\phivec$.
This would require solving
\[\min_{\thetavec} L_1(\thetavec, \aa), \hspace{0.2cm} \min_{\phivec} L_2(\phivec, \aa)\]
where $\aa = (\alpha_1, \ldots, \alpha_N)$,\footnote{We index $\aa$ by $N$ to emphasize that the dimension grows with the sample size. }
\begin{align}\label{Ltheta}
	L_1(\thetavec, \aa) & = \int_0^1 \sum_{i = 1}^N \sum_{t = 1}^{T}{\rho_u(y_{it} - \alpha_i - \xx_{it}^\T\betavec(u \mid \thetavec)) \mathrm{d}u}\\ 
		& = \sum_{i = 1}^N \sum_{t = 1}^{T}{\left\{
			(y_{it} - \alpha_i)(u_{it}(\thetavec, \alpha_i) - 0.5) 
			+ \xx_{it}^\T\thetavec\left[\bar\Bvec - \Bvec(u_{it}(\thetavec, \alpha_i))\right]\right\}, \nonumber
	}
\end{align}
\begin{align}\label{Lphi}
	L_2(\phivec, \aa) & = \int_0^1 \sum_{i = 1}^N {\rho_v(\alpha_i - \zz_i^\T\gammavec(v \mid \phivec)) \mathrm{d}v}\\ 
		& = \sum_{i = 1}^N{\left\{
			\alpha_i(v_i(\phivec, \alpha_i) - 0.5) 
			+ \zz_i^\T\phivec\left[\bar\Cvec - \Cvec(v_i(\phivec, \alpha_i))\right]\right\}. \nonumber
	}
\end{align}
To obtain expressions (\ref{Ltheta})\footnote{
     The expression for $L_1(\thetavec, \aa)$ bears some similarity to \cites{koenker2004} loss function for unpenalized fixed-effects quantile regression, 
	which is defined by $L(\betavec, \aa) = \sum_j \sum_i \sum_t w_j \rho_{u_j}(y_{it} - \alpha_i - \xx_{it}\betavec(u_j))$ 
	and can be seen as a discretized, non-parametrized, and weighted version of $L_1(\thetavec, \aa)$.} 
and (\ref{Lphi}), we used equation (9) from \cite{iqr},
and define
\begin{equation}\label{B}
  \Bvec(u) = \int_0^u \bvec(s) \mathrm{d}s, \hspace{0.3cm} \bar\Bvec = \int_0^1{\Bvec(u)\mathrm{d}u},
\end{equation}
\begin{equation}\label{C}
\Cvec(v) = \int_0^v \cvec(s) \mathrm{d}s, \hspace{0.3cm} \bar\Cvec = \int_0^1{\Cvec(v)\mathrm{d}v}.
\end{equation}
In the formulas, $u_{it}(\thetavec, \alpha_i)$ and $v_i(\phivec, \alpha_i)$ are such that $y_{it} - \alpha_i = \xx_{it}^\T\thetavec\bvec(u_{it}(\thetavec, \alpha_i))$
and $\alpha_i = \zz_{i}\phivec\cvec(v_i(\phivec, \alpha_i))$, respectively. This also implies that
\begin{equation}\label{u}
u_{it}(\thetavec, \alpha_i) = F_{y - \alpha}(y_{it} - \alpha_i \mid \xx_{it}, \thetavec)
\end{equation}
is the cumulative distribution of $Y_{it} - \alpha_i$,
given $\xx_i$, with parameter $\thetavec$; and 
\begin{equation}\label{v}
v_i(\phivec, \alpha_i) = F_{\alpha}(\alpha_i \mid \zz_i, \phivec)
\end{equation} 
is the cumulative distribution 
of $\alpha_i$, given $\zz_i$, with parameter $\phivec$.

In practice, the vector $\aa$ of individual effects is not known 
and must be estimated. We propose estimating $(\thetavec, \phivec, \aa)$ by solving
\begin{equation}\label{L}
\min_{\thetavec, \phivec, \aa} L_1(\thetavec, \aa) + L_2(\phivec, \aa).
\end{equation}
The proposed loss function is similar to that of a penalized fixed-effects estimator
in which $L_2(\phivec, \aa)$ plays the role of a penalty term.
Intuitively, $L_2(\phivec, \aa)$ shrinks the estimated fixed effects towards their assumed conditional distribution,
introducing some degree of smoothing, improving model identification and efficiency, and avoiding overfitting.
At the same time, $L_2(\phivec, \aa)$ carries information on the parameter $\phivec$ that describes the quantile function of $\aa$. 

Since both $\aa$ and $\phivec$ are treated as parameters, this approach combines features
of fixed-effects estimators, which only estimate $\thetavec$ and $\aa$, and random-effects models, which
directly estimate $\thetavec$ and $\phivec$. Computation, however, is much simpler than that 
of purely random-effects methods (e.g., \citealp{kim, arellano}).

The gradient functions of $L(\thetavec, \phivec, \aa) = L_1(\thetavec, \aa) + L_2(\phivec, \aa)$ can be written as
\begin{equation}\label{Gtheta}
\G_{\thetavec}(\thetavec, \aa) = \nabla_{\ve(\thetavec)} L(\thetavec, \phivec, \aa) = 
	\sum_{i = 1}^N \sum_{t = 1}^{T}{\left[\bar\Bvec - \Bvec(u_{it}(\thetavec, \alpha_i))\right] \otimes \xx_{it}},
\end{equation}
\begin{equation}\label{Gphi}
\G_{\phivec}(\phivec, \aa) = \nabla_{\ve(\phivec)} L(\thetavec, \phivec, \aa) = 
	\sum_{i = 1}^N {\left[\bar\Cvec - \Cvec(v_i(\phivec, \alpha_i))\right] \otimes \zz_i},
\end{equation}
\begin{equation}\label{Ga}
G_{\alpha_i}(\alpha_i, \thetavec, \phivec) = \nabla_{\alpha_i} L(\thetavec, \phivec, \aa) = 
	\left[\sum_{t = 1}^{T} {(0.5 - u_{it}(\thetavec, \alpha_i))}\right] + (v_i(\phivec, \alpha_i) - 0.5),
\end{equation}
where $\ve$ denotes the vectorization operator and $\otimes$  the kronecker product.
The model parameters, $(\thetavec, \phivec, \aa)$, only enter equations \eqref{Gtheta}--\eqref{Ga}
through the cumulative distribution functions $u_{it}(\thetavec, \alpha_i)$ and $v_i(\phivec, \alpha_i)$ defined in \eqref{u} and \eqref{v}.
Note that $\G_{\thetavec}(\thetavec, \aa)$ does not carry information on $\phivec$, and
$\G_{\phivec}(\phivec, \aa)$ does not carry information on $\thetavec$; while
$G_{\alpha_i}(\alpha_i, \thetavec, \phivec)$ depends on both $\thetavec$ and $\phivec$.
As shown by \cite{iqr}, $\G_{\thetavec}(\thetavec, \aa)$ and $\G_{\phivec}(\phivec, \aa)$ approach zero when the 
distributions of $u_{it}(\thetavec, \alpha_i)$ and $v_i(\phivec, \alpha_i)$ tend to be uniform. This reflects the data-generating process described in (\ref{themodel}),
which involves the two independent uniform variables $U_{it}$ and $V_i$.

Equation (\ref{Ga}) clarifies the role of the ``penalty'' term $L_2(\phivec, \aa)$: 
\begin{itemize}
\item{the left-hand side of (\ref{Ga}), $\sum_{t = 1}^{T} {(0.5 - u_{it}(\thetavec, \alpha_i))} = \nabla_{\alpha_i}L_1(\thetavec, \aa)$,
is an unpenalized estimating equation for $\alpha_i$. It approaches zero when $u_{i1}(\thetavec, \alpha_i), u_{i2}(\thetavec, \alpha_i), \ldots, u_{iT}(\thetavec, \alpha_i)$ are evenly spaced in $(0,1)$,
imposing a within-cluster uniformity of $u_{it}(\thetavec, \alpha_i)$ which mirrors the assumption of independence between $U_{it}$ and $V_i$;}
\item{the right-hand side, $ (v_i(\phivec, \alpha_i) - 0.5) = \nabla_{\alpha_i}L_2(\phivec, \aa)$, is a penalty term that shrinks 
the value of $\alpha_i$ towards its conditional median, $\zz_{i}^\T\gammavec(0.5 \mid \phivec) = \zz_{i}^\T\phivec\cvec(0.5)$.
}
\end{itemize}
A desirable property of the proposed penalization is that it only affects the estimates of $\aa$ when the clusters are relatively small.
As $T \to \infty$, each cluster contains sufficient information to estimate its own individual effect and,
consistently, the penalty term $(v_i(\phivec, \alpha_i) - 0.5)$ in equation (\ref{Ga}) becomes irrelevant.

Estimation can be performed by the following iterative process: (i) given $\aa$, estimate $\thetavec$
and $\phivec$ separately by solving $\G_{\thetavec}(\thetavec, \aa)=0$ and $\G_{\phivec}(\phivec, \aa)=0$;
(ii) given $(\thetavec, \phivec)$, compute a new estimate of $\aa$ by solving $G_{\alpha_i}(\alpha_i, \thetavec, \phivec) = 0$,
$i = 1, \ldots, N$. Step (i) can be implemented with standard routines available in the \texttt{qrcm} package,
while step (ii) requires finding the zero of $N$ univariate estimating equations. Neither $u_{it}(\thetavec, \alpha_i)$ nor $v_i(\phivec, \alpha_i)$ are generally 
available in closed form, and can be evaluated by using a bisection algorithm. Note that the objective function
defined by (\ref{L}) is a smooth function of all parameters, unlike the loss function of standard quantile regression.

The fact that the quantile function may be ill-defined at some value of the parameters can be an issue 
during estimation. In the implementation of the \texttt{qrcm} package, we use
unconstrained optimization from carefully chosen initial values. 
The algorithm is described in detail in Appendix B.

\subsection{A new family of penalized fixed-effects estimators}\label{subsec:penalty}

A possible interpretation of the proposed loss function,
\begin{equation}\label{Lbis}\nonumber
L(\thetavec, \phivec, \aa) = L_1(\thetavec, \aa) + L_2(\phivec, \aa),
\end{equation}
is to consider $L_2(\phivec, \aa)$ as a penalty term that shrinks the estimated individual effects
towards their conditional median, $\zz_i^\T\gammavec(0.5 \mid \phivec)$. Unlike standard
penalizations, however, $L_2(\phivec, \aa)$ may depend on level-2 covariates and is a function of estimated parameters.

To clarify this idea, consider a more traditional penalized loss function
\begin{equation}\label{Lp}\nonumber
L_{\lambda}(\thetavec, \aa) = L_1(\thetavec, \aa) + \lambda L_2(\aa),
\end{equation}
where $L_2(\aa)$ is a penalty term which does not contain $\phivec$, and $\lambda$ is a tuning parameter.
Common choices of $L_2(\aa)$ are the $\ell_1$-penalization, $L_2(\aa) = \sum_{i = 1}^N|\alpha_i|$,
which was used by \cite{koenker2004} to implement longitudinal quantile regression,
and the $\ell_2$-penalization, $L_2(\aa) = \sum_{i = 1}^N \alpha_i^2$.

Standard $\ell_1$- and $\ell_2$-penalized fixed-effects methods are computationally simple and can
substantially improve efficiency of the estimates of the structural parameters. However, besides the 
fact that they do not allow for estimation of $\phivec$, they present 
some important limitations: 
(i) they do not use prior knowledge on the distribution of the individual effects; 
(ii) they can introduce bias; 
(iii) they apply the same penalization to all clusters; and 
(iv) they require to specify a tuning parameter.

For instance, $\ell_1$-penalized estimators of quantile regression coefficients are asymptotically biased unless
$(\alpha_1, \ldots, \alpha_N)$ are independent and identically distributed with zero median \citep{lamarche}. This is just
a consequence of the $\ell_1$-penalty term being a sum of absolute deviations from zero,
which does not generally reflect the true distribution of $\aa$ and the effect of level-2 covariates on it. 
Moreover, the same value of $\lambda$ is used for all clusters, ignoring the fact that the variance
of the individual effects may differ across subgroups of the population. 

The tuning constant $\lambda$ determines the degree of shrinking and,
in the standard random-intercept linear model, its optimal value is $\sigma^2_{\epsilon}/\sigma^2_{\alpha}$,
i.e., a function of nuisance scale parameters (e.g., \citealp{koenker2004}). 
Outside the restrictive conditions of linear models, not only the choice of $\lambda$ becomes problematic,
but also the use of a single value of $\lambda$ for all clusters is questionable.

The novelty of our approach is that, unlike the $\ell_1$- and $\ell_2$-penalizations, 
the term $L_2(\phivec, \aa)$ reflects the true (conditional) distribution of $\aa$ and carries information about its parameters, $\phivec$. 
Our estimator presents the following advantages over standard penalized methods:
(i) it enables incorporating parametric assumptions on the distribution of $\aa$;
(ii) it permits estimating all parameters consistently; 
(iii) it applies a different degree of shrinking to each cluster, by modeling the effect of level-2 covariates 
on the distribution of the individual effects; and 
(iv) it does not require selecting a tuning constant, as no nuisance parameters are present.

To clarify point (iv), consider the loss function of an $\ell_2$-penalized linear regression model:
$L_{\lambda}(\betavec, \aa) = L_1(\betavec, \aa) + \lambda L_2(\aa) = 
\sum_{i = 1}^N\sum_{t = 1}^T{(y_{it} - \xx_i^\T\betavec - \alpha_i)^2} + \lambda \sum_{i = 1}^N \alpha_i^2$.
Here, $L_1(\betavec, \aa)$ and $L_2(\aa)$ lack information on the nuisance scale parameters
$\sigma^2_{\epsilon} = \text{var}(Y_{it} - \xx_i^\T\betavec - \alpha_i)$ and $\sigma^2_{\alpha} = \text{var}(\alpha_i)$.
This is adjusted for by the tuning constant $\lambda = \sigma^2_{\epsilon}/\sigma^2_{\alpha}$.
In our special type of penalized estimator, instead, $L_1(\thetavec,\aa)$ and $L_2(\phivec,\aa)$ carry information on \textit{all} model parameters.
Intuitively, this means that $L_1(\thetavec,\aa)$ and $L_2(\phivec,\aa)$ are already ``properly scaled''.
The tuning constant can be thought of as an implicit parameter, a function of $\thetavec$ and $\phivec$. 
Although a more general estimator with criterion function $L_1(\thetavec,\aa) + \lambda L_2(\phivec,\aa)$ 
could in principle be formulated, choosing $\lambda = 1$ appears natural and avoids the problem of
selecting the tuning parameter. 




\section{INFERENCE}\label{sec:asy}

The asymptotic properties of fixed-effects estimators are complicated by the fact that,
as $N \to \infty$, the dimension of the parameter $\aa$ 
tends to infinity. Unless $T \to \infty$, 
the individual effects $\alpha_i$ are estimated using a fixed number of observations.
This is often referred to as the ``incidental parameter'' problem 
(\citealp{NS}; \citealp{lancaster}), which causes widely used estimators, such as
maximum likelihood and M-estimators, to be inconsistent. 

To develop the asymptotic theory of our estimator,
we follow the recent panel data literature in econometrics and deal with the incidental parameter problem by considering asymptotic sequences where  both $N$ and $T$ tend to infinity (e.g., \citealp{Phillips:1999p733}, \citealp{hn}, \citealp{koenker2004},  \citealp{fv}, \citealp{ArellanoHahn2007}, \citealp{lamarche}, and \citealp{kato}). Under this  approximation, we show that our estimators are consistent but might have biases in the asymptotic distribution depending on the relative rate of convergence of $N$ and $T$. We apply the theory of M-estimators
(e.g., \citealp{newey}), and use well-established results to handle the following non-standard
features of our problem: (i) the estimators of $\thetavec$, $\phivec$ and $\aa$ converge at
different rates (e.g., \citealp{radchenko}; \citealp{cs}; \citealp{masuda}); and (ii) additional
conditions are required on the relative growth rate of $N$ and $T$ (e.g., 
\citealp{hn}; \citealp{fv}; \citealp{neweynotes}).

Let $\xx_{it} = (\xx_{1i}^\T,\xx_{2it}^\T)^\T$, where $\xx_{1i}$ contains the time-invariant components including the constant and $\xx_{2it}$ contains the time-varying covariates.  We use the following sufficient conditions to establish the identification of the parameters and derive the asymptotic properties of the estimators:

\begin{assumption}[Longitudinal ILM Estimator]\label{ass:ilm}
(i) The data generating process is $Y_{it} = \xx_{it}^{\T}\thetavec^0\bvec(U_{it}) + \zz_i^{\T}\phivec^0\cvec(V_i)$, where $\plim_{N\to \infty}N^{-1} \sum_{i=1}^N \alpha_i^0 =0$ for $\alpha_i^0 := \zz_i^{\T}\phivec^0\cvec(V_i)$ and, conditional on $\{(\xx_{it}, \zz_i) : 1 \leq i \leq N, 1 \leq t \leq T \}$,   $U_{it} \sim U(0,1)$ independently over $i$ and $t$,  $V_i \sim U(0,1)$ independently over $i$, and $U_{it}$ and $V_j$ are independent over all $i,t,j$. (ii) For each $i$, $\plim_{T\to \infty} T^{-1} \sum_{t=1}^T \bar \xx_{2it} \bar \xx_{2it}^{\T}$ exists and is positive definite for $\bar \xx_{2it} = (1,\xx_{2it}^\T)^\T$, and $\plim_{N\to \infty} N^{-1} \sum_{i=1}^N \xx_{1i}  \xx_{1i}^{\T}$ and $\plim_{N\to \infty}N^{-1} \sum_{i=1}^N \zz_{i}  \zz_{i}^{\T}$ exist and are positive definite. (iii) The variables $\xx_{it}^{\T}\thetavec^0\bvec(U_{it})$ and $\zz_i^{\T}\phivec^0\cvec(V_i)$ are finite a.s., and $\xx_{it}^{\T}\thetavec^0 \bvec'(U_{it})$ and $\zz_{i}^{\T}\phivec^0 \cvec'(V_{i})$ are positive a.s.   (iv)   There exist two sets of quantile indexes $\{u_1, \ldots, u_{d_{b}} \}$ and $\{v_1, \ldots, v_{d_{c}} \}$ such that the matrices $\bvec(u_1, \ldots, u_{d_{b}}) := [\bvec(u_1), \ldots, \bvec(u_{d_{b}})]$ and $\cvec(v_1, \ldots, v_{d_{c}}) := [\cvec(v_1), \ldots, \cvec(v_{d_{c}})]$ have full rank.
(v) The functions $u \mapsto \bvec(u)$  and $v \mapsto \cvec(v)$ are three times continuously differentiable on $(0,1)$, $\sup_i \plim_{T \to \infty} T^{-1} \sum_{t=1}^T \| \xx_{2it} \|^\xi < \infty$ for some $\xi > 32$, and $\|(\xx_{1i}^\T, \zz_{i}^\T)^\T\|$ is bounded a.s. (vi) The following probability limits exist:
\begin{eqnarray*}
\bar{\H}_{\thetavec} &=&  \plim_{N,T \to \infty} \frac{1}{NT} \sum_{i=1}^N \sum_{t=1}^T \Ep\left( \H_{\thetavec it}\right), \\ 
\bar{\H}_{\phivec} &=&  \plim_{N \to \infty} \frac{1}{N} \sum_{i=1}^N  \Ep\left( \frac{\cvec(V_{i}) \cvec(V_{i})^{\T}  }{\zz_{i}^{\T}\phivec^0 \cvec'(V_{i})} \otimes  \zz_{i}\zz_{i}^{\T} \right), \\
\bar{\H}_{\phivec\thetavec} &=&   \plim_{N,T \to \infty} \frac{1}{N} \sum_{i=1}^N  \Ep\left(\frac{\cvec(V_{i})  }{\zz_{i}^{\T}\phivec^0 \cvec'(V_{i})}  \otimes \zz_{i} \right)  \frac{\sigma_i^2}{T} \sum_{t=1}^T \Ep\left( \frac{\bvec(U_{it})  }{\xx_{it}^{\T}\thetavec^0 \bvec'(U_{it})}  \otimes \xx_{it}  \right)^{\T}, \\
\bar{\bvec}_{\thetavec} &=&   \plim_{N,T \to \infty} \frac{1}{NT} \sum_{i=1}^N \sum_{t=1}^T \Ep(\bvec_{\thetavec it}), \ \ \bar{\bvec}_{\phivec} =   \plim_{N \to \infty} \frac{1}{N} \sum_{i=1}^N \Ep(\bvec_{\phivec i}) , \\
\bar{\O}_{\thetavec} &=& \plim_{N,T \to \infty} \frac{1}{NT} \sum_{i=1}^N  \sum_{t=1}^T \Ep(\varphivec_{\thetavec it} \varphivec_{\thetavec it}^{\T}), \\ 
\bar{\O}_{\phivec} &=& \plim_{N \to \infty} \frac{1}{N} \sum_{i=1}^N \Ep\left\{\left[\bar \Cvec - \Cvec(V_i)\right]\left[\bar \Cvec - \Cvec(V_i)\right]^{\T} \otimes \zz_{i}\zz_{i}^{\T}\right\},
\end{eqnarray*}
where the expectation $\Ep$ is taken with respect to the distribution of $U_{it}$ and $V_i$, $\otimes$ denotes the Kronecker product, $\cvec'$ and $\cvec''$ denote the vectors of first and second  derivatives of  $v \mapsto \cvec(v)$, 
\begin{eqnarray*}
\H_{\thetavec it} &=&  \frac{\bvec(U_{it}) \bvec(U_{it})^{\T}  }{\xx_{it}^{\T}\thetavec^0 \bvec'(U_{it})} \otimes   \xx_{it}  \xx_{it}^{\T}  -  \left( \frac{\bvec(U_{it})  }{\xx_{it}^{\T}\thetavec^0 \bvec'(U_{it})}  \otimes \xx_{it}  \right) \frac{\sigma_i^2}{T} \sum_{t=1}^T \Ep\left( \frac{\bvec(U_{it})  }{\xx_{it}^{\T}\thetavec^0 \bvec'(U_{it})}  \otimes \xx_{it}  \right)^{\T} ,\\
\bvec_{\phivec i} &=& \left\{ \frac{[\sigma_i^2(V_i - 0.5) - \beta_i]\cvec(V_{i})  }{\zz_{i}^{\T}\phivec^0 \cvec'(V_{i})}  +  \frac{\sigma_i^4 \cvec(V_{i})  \zz_{i}^{\T}\phivec^0 \cvec''(V_{i})}{24[\zz_{i}^{\T}\phivec^0 \cvec'(V_{i})]^3} - \frac{\sigma_i^4 \cvec'(V_{i})  }{24[\zz_{i}^{\T}\phivec^0 \cvec'(V_{i})]^2}\right\}\otimes \zz_{i},\\
\bvec_{\thetavec it} &=&   \left\{ \frac{[ \sigma_i^2(U_{it} - 0.5)+\beta_i] \bvec(U_{it})  }{\xx_{it}^{\T}\thetavec^0 \bvec'(U_{it})}  -  \frac{\sigma_i^4 \bvec(U_{it})  \xx_{it}^{\T}\thetavec^0 \bvec''(U_{it})}{24[\xx_{it}^{\T}\thetavec^0 \bvec'(U_{it})]^3} - \frac{\sigma_i^4 \bvec'(U_{it})  }{24[\xx_{it}^{\T}\thetavec^0 \bvec'(U_{it})]^2}\right\} \otimes \xx_{it}, \\
\varphivec_{\thetavec it} &=& [\bar \Bvec - \Bvec(U_{it})] \otimes \xx_{it} + \frac{\sigma_i^2}{T} \sum_{s=1}^T \Ep\left( \frac{\bvec(U_{is})  }{\xx_{is}^{\T}\thetavec^0 \bvec'(U_{is})}  \otimes \xx_{is}  \right) (U_{it} - 0.5),\\
\beta_i &=&  \sigma_i^4 \plim_{T\to \infty} \frac{1}{T} \sum_{t=1}^T \Ep \left(\frac{U_{it}  - 0.5}{\xx_{it}^{\T}\thetavec^0 \bvec'(U_{it})} -  \frac{\sigma_i^2 \xx_{it}^{\T}\thetavec^0 \bvec''(U_{it})}{24[\xx_{it}^{\T}\thetavec^0 \bvec'(U_{it})]^3}\right),\\
 \sigma_i^2 &=&  \left[ \plim_{T\to \infty} \frac{1}{T} \sum_{t=1}^T  \Ep\left( \frac{1  }{\xx_{it}^{\T}\thetavec^0 \bvec'(U_{it})} \right) \right]^{-1},
\end{eqnarray*}
and $\bvec'$ and $\bvec''$ denote the vectors of first and second derivatives of  $u \mapsto \bvec(u)$. 
(vii) The minimum eigenvalues of the matrices $\bar{\H}_{\thetavec}$ and $\bar{\H}_{\phivec}$ are bounded away from zero, and $\sup_i \sigma_i^2 < c$ for some constant $c > 0$. 
\end{assumption}

We use Assumptions \ref{ass:ilm}(i)-(iv) to establish the identification of all the model parameters.  Assumption \ref{ass:ilm}(i) imposes that the model is correctly specified. It also requires $\alpha_i$ and $Y_{it} - \alpha_i$ to be conditionally 
independent across $(i,t)$, and to be independent of each other. We do not impose any sampling condition on the covariate sequences $\{(\xx_{it}, \zz_i) : 1 \leq i \leq N, 1 \leq t \leq T \}$, other than the existence of some limits.  Assumptions \ref{ass:ilm}(ii)-(iii) apply standard regularity conditions for parameter  identification in quantile regression to our longitudinal model (e.g., \citealp{ACF06}). For example, Assumption \ref{ass:ilm}(iii) imposes that the conditional quantile and density functions of $Y_{it} - \alpha_i$ and $\alpha_i$ are bounded. These conditions, together with a location normalization on the fixed effects in Assumption \ref{ass:ilm}(i), guarantee that $\betavec(\cdot)$ and $\gammavec(\cdot)$ in the model \eqref{themodel} are identified.\footnote{We normalize  the mean of the fixed effects. Alternative  normalizations on the median or other quantile of the fixed effects are also possible.}  Then, Assumption \ref{ass:ilm}(iv)  pins down $\thetavec^0$ and $\phivec^0$ from the system of linear equations $\betavec(\cdot) = \thetavec^0 \bvec(\cdot)$ and $\gammavec(\cdot) = \phivec^0 \cvec(\cdot)$. Assumption \ref{ass:ilm}(iv) provides a sufficient condition to guarantee  existence and uniqueness of solution to the system from a subset of equations, which  is easy to verify in practice. It can be replaced by any other existence and uniqueness condition.

Assumptions \ref{ass:ilm}(v)-(vii)  impose regularity conditions to derive the distribution of the estimators in large samples. The derivation relies on a general asymptotic expansion for fixed effects M-estimators given in Appendix A, which extend the results of \cite{hn} and \cite{fv} to estimators with mixed-rates asymptotics.  Assumption \ref{ass:ilm}(v) requires sufficient smoothness and bounded moments of the objective functions \eqref{Ltheta} and \eqref{Lphi} and their partial derivatives, which are needed to carry out  higher-order  expansions of these functions. Assumption \ref{ass:ilm}(vi) guarantees that all the terms of the  expansions are well-defined. Finally, Assumption   \ref{ass:ilm}(vii) is a standard condition imposing that the limit Hessian matrices of the objective functions are non-singular.

\vspace{0.4cm}


\begin{theorem}\label{thm:ilm} Suppose that Assumption \ref{ass:ilm} holds. Then, (i) $\thetavec^0$ and $\phivec^0$ are identified. (ii) If $N = O(T)$,
$$
\sqrt{NT}\left( \ve\left[\hat{\thetavec} - \thetavec^0\right] + \frac{\bar{\H}_{\thetavec}^{-1} \bar{\bvec}_{\thetavec}}{T} \right) \to_d \bar{\H}_{\thetavec}^{-1} \text{N}(0, \bar{\O}_{\thetavec}). 
$$ 
(iii) If $N = O(T^2)$,
$$
\sqrt{N}\left( \ve\left[\hat{\phivec} - \phivec^0\right] + \frac{\bar{\H}_{\phivec}^{-1} (\bar{\bvec}_{\phivec} +  \bar{\H}_{\phivec\thetavec} \bar{\H}_{\thetavec}^{-1}\bar{\bvec}_{\thetavec})}{T} \right) \to_d \bar{\H}_{\phivec}^{-1} \text{N}(0, \bar{\O}_{\phivec}). 
$$ 
The expressions of all the terms are given in Assumption \ref{ass:ilm}.
\end{theorem}

\vspace{0.4cm}

Theorem \ref{thm:ilm} shows that the parameters $\thetavec^0$ and $\phivec^0$ are identified and  their estimators $\hat{\thetavec}$ and $\hat{\phivec}$ have a normal distribution in large samples with different rates of convergence. The large sample distribution of the plugin estimators of $\betavec^0(\cdot) = \thetavec^0\bvec(\cdot)$ and $\gammavec^0(\cdot) = \phivec^0 \cvec(\cdot)$ can be obtained by the delta method. Let $\hat{\betavec}(u) = \hat{\thetavec} \bvec(u)$ and $\hat{\gammavec}(v) = \hat{\phivec} \cvec(v)$, for any $u,v \in (0,1)$. Then, if $N = O(T)$,
$$
\sqrt{NT}\left( \hat{\betavec}(u) - \betavec^0(u) + \frac{\ve_{d_x,d_{b}}^{-1}(\bar{\H}_{\thetavec}^{-1} \bar{\bvec}_{\thetavec}) \bvec(u)}{T} \right) \to_d  \text{N}(0, \tilde{\bvec}(u)^{\T}\bar{\H}_{\thetavec}^{-1}\bar{\O}_{\thetavec}\bar{\H}_{\thetavec}^{-1} \tilde{\bvec}(u)),
$$
where $\ve_{d,k}^{-1}$ is the inverse vectorization operator that maps a $dk$-vector to a $d \times k$ matrix, i.e. $\ve_{d,k}^{-1}(\vvec) = \{[\ve(\I_k)]^{\T} \otimes \I_d \} (\I_k \otimes \vvec)$, $\I_n$ is the identity matrix of size $n$, and $\tilde{\bvec}(u) = \bvec(u) \otimes \I_{d_x}$. Similarly, if $N = O(T^2)$,
\begin{multline*}
\sqrt{N}\left( \hat{\gammavec}(v) - \gammavec^0(v) + \frac{\ve_{d_z,d_{c}}^{-1}(\bar{\H}_{\phivec}^{-1} (\bar{\bvec}_{\phivec} +  \bar{\H}_{\phivec\thetavec} \bar{\H}_{\thetavec}^{-1}\bar{\bvec}_{\thetavec})) \cvec(v)}{T} \right) \\ \to_d  \text{N}(0, \tilde{\cvec}(v)^{\T}\bar{\H}_{\phivec}^{-1}\bar{\O}_{\phivec}\bar{\H}_{\phivec}^{-1} \tilde{\cvec}(v)),
\end{multline*}
where $\tilde{\cvec}(v) = \cvec(v) \otimes \I_{d_z}$. 

The rates of convergence of all the estimators agree with the square roots of the dimensions of the observations that are informative about the corresponding parameters. Thus, the rate is $\sqrt{NT}$ for $\thetavec^0$ and $\betavec^0(u)$, and  $\sqrt{N}$ for $\phivec^0$ and $\gammavec^0(v)$. All the estimators might suffer from  bias in short panels due to the estimation of the fixed effects. The order of this bias is the inverse of the number of observations that are informative about each fixed effect, i.e. $T^{-1}$.  Comparing the rates of convergence with the order of the bias,  we can see that the biases of $\hat{\thetavec}$ and $\hat{\betavec}(u)$ are negligible in the asymptotic distribution when $N/T \to 0$, whereas the biases of $\hat{\phivec}$ and $\hat{\gammavec}(v)$ are negligible when $N/T^2 \to 0$. These biases can be reduced by using analytical or jackknife corrections (e.g.,  \citealp{hn}, \citealp{fv}, or \citealp{dj15}). We provide consistent analytical estimators of the components of the biases and variances below. 


%

We construct estimators of the components of the asymptotic distribution using sample analogs evaluated at the estimated value of the parameters,  e.g., $\hat{u}_{it} = u_{it}(\hat{\thetavec}, \hat \alpha_i)$ and $\hat{v}_{i} = v_{i}(\hat{\phivec}, \hat \alpha_i)$.  Then,
\begin{eqnarray*}
\hat{\H}_{\thetavec} &=&   \sum_{i=1}^N \sum_{t=1}^T \hat{\H}_{\thetavec it}, \ \ \hat{\H}_{\phivec} =   \frac{1}{N} \sum_{i=1}^N  \left( \frac{\cvec(\hat v_{i}) \cvec(\hat v_{i})^{\T}  }{\zz_{i}^{\T}\phivec^0 \cvec'(\hat v_{i})} \otimes  \zz_{i}\zz_{i}^{\T} \right), \\
\hat{\H}_{\phivec\thetavec} &=&  \frac{1}{N} \sum_{i=1}^N  \left(\frac{\cvec(\hat v_{i})  }{\zz_{i}^{\T}\hat \phivec \cvec'(\hat v_{i})}  \otimes \zz_{i} \right)  \frac{\hat \sigma_i^2}{T} \sum_{t=1}^T \left( \frac{\bvec(\hat u_{it}) }{\xx_{it}^{\T}\hat \thetavec \bvec'(\hat  u_{it})}  \otimes \xx_{it}  \right)^{\T}, \\
\hat{\bvec}_{\thetavec} &=&  \frac{1}{NT} \sum_{i=1}^N \sum_{t=1}^T \hat \bvec_{\thetavec it}, \ \ \hat{\bvec}_{\phivec} =  \frac{1}{N} \sum_{i=1}^N  \hat \bvec_{\phivec i}, \\
\hat{\O}_{\thetavec} &=&  \frac{1}{NT} \sum_{i=1}^N  \sum_{t=1}^T \hat  \varphivec_{\thetavec it} \hat \varphivec_{\thetavec it}^{\T}, \ \ 
\hat{\O}_{\phivec} = \frac{1}{N} \sum_{i=1}^N \left\{\left[\bar \Cvec - \Cvec(\hat v_i)\right]\left[\bar \Cvec - \Cvec(\hat v_i)\right]^{\T} \otimes \zz_{i}\zz_{i}^{\T}\right\}, 
\end{eqnarray*}
where 
\begin{eqnarray*}
\hat \H_{\thetavec it} &=&  \frac{\bvec(\hat u_{it}) \bvec(\hat u_{it})^{\T}  }{\xx_{it}^{\T}\hat \thetavec \bvec'(\hat u_{it})} \otimes   \xx_{it}  \xx_{it}^{\T}  -  \left( \frac{\bvec(\hat u_{it})  }{\xx_{it}^{\T}\hat \thetavec \bvec'(\hat u_{it})}  \otimes \xx_{it}  \right) \frac{\hat \sigma_i^2}{T} \sum_{t=1}^T \left( \frac{\bvec(\hat u_{it})  }{\xx_{it}^{\T}\hat \thetavec \bvec'(\hat u_{it})}  \otimes \xx_{it}  \right)^{\T} ,\\
\hat \H_{\phivec i} &=&  \frac{\cvec(\hat v_{i}) \cvec(\hat v_{i})^{\T}  }{\zz_{i}^{\T}\phivec^0 \cvec'(\hat v_{i})} \otimes  \zz_{i}\zz_{i}^{\T}  -  \left( \frac{\cvec(\hat v_{i})  }{\zz_{i}^{\T}\hat \phivec \cvec'(\hat v_{i})}  \otimes \zz_{i}  \right) \frac{\hat \sigma_i^2}{T} \left( \frac{\cvec(\hat v_{i})  }{\zz_{i}^{\T}\hat \phivec \cvec'(\hat v_{i})}  \otimes \zz_{i}  \right)^{\T} ,\\
\hat  \bvec_{\thetavec it} &=&   \left\{ \frac{[\hat \sigma_i^2(\hat u_{it} - 0.5) + \hat \beta_i ] \bvec(\hat u_{it})  }{\xx_{it}^{\T}\hat \thetavec \bvec'(\hat u_{it})}  -  \frac{\hat \sigma_i^4 \bvec(\hat u_{it})  \xx_{it}^{\T}\thetavec^0 \bvec''(\hat u_{it})}{24[\xx_{it}^{\T}\hat \thetavec \bvec'(\hat u_{it})]^3} - \frac{\hat \sigma_i^4 \bvec'(\hat u_{it})  }{24[\xx_{it}^{\T}\hat \thetavec \bvec'(\hat u_{it})]^2}\right\} \otimes \xx_{it}, \\
\hat \bvec_{\phivec i} &=& \left\{ \frac{[\hat \sigma_i^2(\hat v_i - 0.5) - \hat \beta_i ]\cvec(\hat v_{i})  }{\zz_{i}^{\T}\hat \phivec \cvec'(\hat v_{i})}  +  \frac{\hat \sigma_i^4 \cvec(\hat v_{i})  \zz_{i}^{\T}\hat \phivec \cvec''(\hat v_{i})}{24[\zz_{i}^{\T}\hat \phivec \cvec'(\hat v_{i})]^3} - \frac{\hat \sigma_i^4 \cvec'(\hat v_{i})  }{24[\zz_{i}^{\T}\hat \phivec^0 \cvec'(\hat v_{i})]^2}\right\}\otimes \zz_{i},\\
\hat \varphivec_{\thetavec it} &=& [\bar \Bvec - \Bvec(\hat u_{it})] \otimes \xx_{it} + \frac{\hat \sigma_i^2}{T} \sum_{s=1}^T \left( \frac{\bvec(\hat u_{is})  }{\xx_{is}^{\T}\hat \thetavec \bvec'(\hat u_{is})}  \otimes \xx_{is}  \right) (\hat u_{it} - 0.5),\\
\hat \beta_i &=&  \hat \sigma_i^4 \frac{1}{T} \sum_{t=1}^T  \left(\frac{\hat u_{it}  - 0.5}{\xx_{it}^{\T}\hat \thetavec \bvec'(\hat u_{it})} -  \frac{\hat \sigma_i^2 \xx_{it}^{\T}\hat \thetavec \bvec''(\hat u_{it})}{24[\xx_{it}^{\T}\hat \thetavec \bvec'(\hat u_{it})]^3}\right),\ \
 \hat  \sigma_i^2 =   \left[ \frac{1}{T} \sum_{t=1}^T   \frac{1  }{\xx_{it}^{\T}\hat \thetavec \bvec'(\hat u_{it})} \right]^{-1}.
\end{eqnarray*}
The consistency of these estimators follows from the law of large numbers and consistency of $\hat \thetavec$, $\hat \phivec$ and $\hat \alpha_i$, $1 \leq i \leq N$, together with the continuous mapping theorem, as all the components are continuous functions of the parameters. 

\section{GOODNESS-OF-FIT MEASURES AND MODEL SELECTION}\label{sec:gof}

\subsection{Goodness-of-fit}\label{subsec:gof}
To assess the model fit, we use the fact that, under the true model, $\hat u_{it}$ and $\hat v_i$ consistently 
estimate the realizations of the two independent uniform variables $U_{it}$ and $V_i$ that generated the data.

A graphical inspection of the joint and marginal distributions of $\hat u_{it}$ and $\hat v_i$
is always recommended. For a formal test, we suggest comparing the empirical distribution of $\hat u_{it}$ and $\hat v_i$, 
$\hat F_{\hat u \hat v}$, with the distribution of $U_{it}$ and $V_i$, given by
\[F_{uv}(u,v) = uv.\]
Following \citeauthor{iqr} (\citeyear{iqr, ctiqr}), we compute a p-value for the null hypothesis
$H_0:$ \{\textit{the model is correct}\} using a Monte Carlo procedure:
\begin{itemize}
\item{\textit{step 0}: compute a test statistic $D$ that measures the distance between 
$\hat F_{\hat u \hat v}(\hat u_{it}, \hat v_i)$ and $F_{uv}(\hat u_{it}, \hat v_i) = \hat u_{it}\hat v_i$;}
\item{\textit{step 1}: simulate new data as $Y_{it}^* = \xx_{it}^{\T}\hat\thetavec\bvec(U_{it}^*) + \zz_i^{\T}\hat\phivec\cvec(V_i^*)$
by randomly generating $(U_{it}^*, V_i^*)$ from two independent $\text{U}(0,1)$ distributions;}
\item{\textit{step 2}: fit the model on the simulated data and compute the corresponding value $D^*$ of the test statistic.}
\end{itemize}
After repeating steps 1-2 for a sufficient number of times, the p-value is computed as the empirical
proportion of cases in which $D^* > D$. In step 1, it is also possible to take a random sample of clusters, 
and to resample the covariates' value within each cluster. To assess local fit, the test could be repeated within subsets of
the original sample identified by specific values of the covariates.

In the implementation of the \texttt{qrcm} package, 
we chose $D$ to be the Kolmogorov-Smirnov statistic, $\sup_{u,v} |\hat F_{\hat u \hat v}(u, v) - uv|$.
This testing procedure is usually reliable, as indicated by the simulation results reported in Section \ref{sec:sim}.

\subsection{Model selection}\label{subsec:ms}
As suggested in Section \ref{subsec:add} and exemplified in the real-data example presented in Section \ref{sec:app}, 
it is usually possible to identify numerous alternative models that have a similar fit
and are not rejected by a goodness-of-fit test. This can be explained by the fact that the same coefficient
functions can be well approximated by different parametric functions.

Important criteria for model selection include parsimony, flexibility, and interpretability.
Nested models can be compared by standard Wald test. Let $\aahat = (\hat \alpha_1, \ldots, \hat \alpha_N)$. To compare non-nested models,
the value of $L_1(\hat\thetavec,\aahat)$ and $L_2(\hat\phivec,\aahat)$ can be used to construct
information criteria such as the \textsc{AIC} \citep{akaike} and the \textsc{BIC} \citep{schwarz}.
These criteria were initially designed for likelihood-based estimators, but can be extended to
estimators defined by the minimizer of a loss function.
For example, a modification of \textsc{BIC} criterion for M-estimators has been described 
by \cite{machado}, while \cite{koenker} used \textsc{AIC}
to compare quantile regression models. 
Consider the probability density function of the asymmetric Laplace distribution,
$$f_p(y \mid \mu,\sigma) = \frac{p(1 - p)}{\sigma(p)}\exp\Bigl\{-\frac{\rho(y - \mu(p))}{2\sigma(p)}\Bigr\},$$
where $\mu(p)$ is a location parameter and corresponds to the $p$-th quantile of the distribution,
while $\sigma(p)$ is a scale parameter.
Although this distribution is not generally considered a plausible model, its 
log-likelihood has been used by numerous authors, including \cite{km1999}
and \cite{lee}, to obtain measures of goodness-of-fit for quantile regression.
Simple algebra permits showing that 
\begin{eqnarray}
\hat\thetavec &=&\argmin_{\thetavec} {\int_0^1{
   \sum_{i = 1}^N\sum_{t = 1}^T 
  \log f_u(y_{it} - \hat\alpha_i \mid \xx_{it}^\T\thetavec\bvec(u), \sigma_1) \mathrm{d}u}}, \nonumber\\\nonumber
\hat\phivec &=& \argmin_{\phivec} {\int_0^1{
   \sum_{i = 1}^N 
  \log f_v(\hat\alpha_i \mid \zz_{i}^\T\phivec\cvec(v), \sigma_2) \mathrm{d}v}},
\end{eqnarray}
i.e., $\hat\thetavec$ and $\hat\phivec$ minimize an ``average'' Laplace log-likelihood,  
in which $u$ and $v$ have been integrated away. 
After substituting $\hat \sigma_1 = L_1(\hat\thetavec, \aahat)/(2NT)$ and $\hat \sigma_2 = L_2(\hat\phivec, \aahat)/(2N)$, 
we obtain the following \textsc{AIC} and \textsc{BIC}:
$$
\textsc{AIC}_1 = \log L_1(\hat\thetavec, \aahat) + \frac{q_1}{NT}, \hspace{0.2cm} \textsc{BIC}_1  = \log L_1(\hat\thetavec, \aahat) + \frac{q_1\log{(NT)}}{2NT},
$$
$$
\textsc{AIC}_2 = \log L_2(\hat\phivec, \aahat) + \frac{q_2}{N}, \hspace{0.2cm} \textsc{BIC}_2  = \log L_2(\hat\phivec, \aahat) + \frac{q_2\log{(N)}}{2N}
$$
where $q_1$ and $q_2$ are the number of non-zero elements of $\thetavec$ and $\phivec$, respectively.
Note that \textsc{BIC}$_1$ can be obtained from equation 2.3 of \cite{lee} by replacing the loss of standard quantile regression
with $L_1(\hat\thetavec, \aahat)$. 

The proposed criteria seem to work well in simulation (see Appendix C). However, they often tend to reward parsimony, possibly sacrificing goodness of fit.
The testing procedure described in Section \ref{subsec:gof} should always be used to perform a preliminary screening of the candidate models.


\section{SIMULATION RESULTS}\label{sec:sim}
We analyze the performance of the estimators $\hat{\betavec}(u)$ and $\hat{\gammavec}(v)$ in finite samples through numerical simulations. In particular, we report the biases and standard errors of these estimators for different values of the dimensions $T$ and $N$ and the orders of the quantiles $u$ and $v$. We also evaluate the empirical size and power of the goodness of fit test.

%

We used the following design to generate the data:
\[Y_{it} = \beta_0(U_{it}) + \beta_1(U_{it})x_{it} + \gamma_0(V_i) + \gamma_1(V_i)z_i\]
where $x_{it} \sim \text{Beta}(2,2)$ and $z_i \sim \text{U}(0,3)$.  In simulation 1, we defined:
\begin{eqnarray}
\beta_0(u) = 1 - 0.5\log(1 - u), & & \beta_1(u) = 1 + 10(u - 0.5)^3, \nonumber\\ 
\gamma_0(v) = 0.5\zeta(v), & & \gamma_1(v) = 0.5\zeta(v), \nonumber
\end{eqnarray}
where $\zeta(v)$ is the quantile function of a standard normal distribution. In simulation 2, we defined: 
\begin{eqnarray}
\beta_0(u) = 2(1 - (1 - u)^{1/4}), & & \beta_1(u) = 3(1 + u), \nonumber\\
\gamma_0(v) = \log(1 - \log(1 - v)), & & \gamma_1(v) = 0.5\log(1 - \log(1 - v)). \nonumber
\end{eqnarray}

To fit the true model, we used $\bvec(u) = \left[1, -\log(1 - u), (u - 0.5)^3\right]^{\T}$ and $\cvec(v) = \left[\zeta(v)\right]$
in simulation 1, and $\bvec(u) = \left[1, 1 - (1 - u)^{1/4}, u\right]^{\T}$ and $\cvec(v) = \left[\log(1 - \log(1 - v))\right]$ in simulation~2.
We ran $R = 1000$ Monte Carlo simulations, with $N = \{150, 300\}$ and $T = \{5, 10\}$. 
In Tables \ref{sim150} and \ref{sim300}, we report the true value of $\betavec(\cdot)$ and $\gammavec(\cdot)$ at the quintiles,
their average estimates, the empirical standard errors across simulations, and the average estimates of the asymptotic standard errors.
Despite the incidental parameters problem, a small bias was found, even with small values of $T$. Also, as $T$ increased,
the observed bias decreased rapidly as predicted by the asymptotic theory of Section \ref{sec:asy}. The estimated standard errors were, on average, very close to their true values.

To assess the performance of the goodness-of-fit procedure described in Section \ref{subsec:gof}, we selected two
nominal significance levels, $\alpha = 0.05$ and $\alpha = 0.10$, and computed the empirical probability
of type I error ($\tilde{\alpha}$) and the power ($1 - \tilde{\beta}$) of the Kolmogorov-Smirnov goodness-of-fit test described in Section \ref{sec:gof}.
The power was estimated by the empirical probability to reject a misspecified model in which the quantile function
was described by an incorrect basis function. In simulation 1, we incorrectly parametrized
$\beta_1(u)$ as a linear function, $\beta_1(u) = \theta_{01} + \theta_{11}u$. In simulation 2, we incorrectly assumed that the individual effects have
a logistic distribution, defined by $\cvec(v) = \left[\log(v/(1 - v))\right]$. Results are shown in the bottom rows of Tables \ref{sim150} and \ref{sim300}.
The risk of type I error was very close to its nominal level, and approached it as the value of $T$ increased. With small values of $N$ and $T$,
the risk of type II error was relatively large, and the power was often less than $50\%$.
However, with $N = 300$ and $T = 10$, and a nominal level of $0.10$ for rejection, the incorrect models were rejected in more than $90\%$ of cases in both scenarios. 

Additional simulation results are reported in Appendix C, where we compare our estimator with \cites{koenker2004}
penalized fixed-effects quantile regression, and discuss the performance of the model selection criteria
presented in Section \ref{subsec:ms}.



\begin{sidewaystable}
\footnotesize
\caption{Simulation results with $N = 150$}
\label{sim150}
\centering
\begin{tabular}{cc
m{0.01cm}m{0.45cm}m{0.45cm}m{0.2cm}m{0.2cm}
m{0.01cm}m{0.45cm}m{0.45cm}m{0.2cm}m{0.2cm}
m{0.01cm}m{0.45cm}m{0.45cm}m{0.2cm}m{0.2cm}
m{0.01cm}m{0.45cm}m{0.45cm}m{0.2cm}m{0.2cm}
}
\hline
\hline
&&& \multicolumn{9}{c}{Simulation 1} & & \multicolumn{9}{c}{Simulation 2}\\
\cline{4-12}\cline{14-22}
$T = 5$ & $u$ && 
	$\beta_0$ & $\hat\beta_0$ & $\text{se}$ & $\hat{\text{se}}$ &&
	$\beta_1$ & $\hat\beta_1$ & $\text{se}$ & $\hat{\text{se}}$ &&
	$\beta_0$ & $\hat\beta_0$ & $\text{se}$ & $\hat{\text{se}}$ &&
	$\beta_1$ & $\hat\beta_1$ & $\text{se}$ & $\hat{\text{se}}$\\
\cline{4-7}\cline{9-12}\cline{14-17}\cline{19-22}
& $0.2$ && 1.11 & 1.17 & .10 & .11	&& 0.73 & 0.72 & .10 & .11	&& 0.11 & 0.11 & .01 & .01	&& 3.60 & 3.79 & .11 & .12\\
& $0.4$ && 1.26 & 1.30 & .11 & .11	&& 0.99 & 0.97 & .11 & .11	&& 0.24 & 0.24 & .03 & .03	&& 4.20 & 4.28 & .12 & .12\\
& $0.6$ && 1.46 & 1.48 & .11 & .11	&& 1.01 & 0.99 & .11 & .11	&& 0.41 & 0.41 & .05 & .05	&& 4.80 & 4.77 & .15 & .15\\
& $0.8$ && 1.80 & 1.79 & .12 & .12	&& 1.27 & 1.24 & .12 & .13	&& 0.66 & 0.66 & .08 & .09	&& 5.40 & 5.26 & .19 & .19\\
& $v$ && 
	$\gamma_0$ & $\hat\gamma_0$ & $\text{se}$ & $\hat{\text{se}}$ &&
	$\gamma_1$ & $\hat\gamma_1$ & $\text{se}$ & $\hat{\text{se}}$ &&
	$\gamma_0$ & $\hat\gamma_0$ & $\text{se}$ & $\hat{\text{se}}$ &&
	$\gamma_1$ & $\hat\gamma_1$ & $\text{se}$ & $\hat{\text{se}}$\\
\cline{4-7}\cline{9-12}\cline{14-17}\cline{19-22}
& $0.2$ && 0.42 & 0.37 & .09 & .10	&& 0.42 & 0.40 & .08 & .08	&& 0.20 & 0.20 & .04 & .04	&& 0.10 & 0.10 & .02 & .03\\
& $0.4$ && 0.13 & 0.11 & .03 & .03	&& 0.13 & 0.12 & .02 & .02	&& 0.41 & 0.41 & .07 & .08	&& 0.21 & 0.20 & .05 & .05\\
& $0.6$ && 0.13 & 0.11 & .03 & .03	&& 0.13 & 0.12 & .02 & .02	&& 0.65 & 0.64 & .11 & .13	&& 0.33 & 0.32 & .08 & .09\\
& $0.8$ && 0.42 & 0.37 & .09 & .10	&& 0.42 & 0.40 & .08 & .08	&& 0.96 & 0.95 & .17 & .19	&& 0.48 & 0.47 & .12 & .13\\
$T = 10$ & $u$ && 
	$\beta_0$ & $\hat\beta_0$ & $\text{se}$ & $\hat{\text{se}}$ &&
	$\beta_1$ & $\hat\beta_1$ & $\text{se}$ & $\hat{\text{se}}$ &&
	$\beta_0$ & $\hat\beta_0$ & $\text{se}$ & $\hat{\text{se}}$ &&
	$\beta_1$ & $\hat\beta_1$ & $\text{se}$ & $\hat{\text{se}}$\\
\cline{4-7}\cline{9-12}\cline{14-17}\cline{19-22}
& $0.2$ && 1.11 & 1.14 & .10 & .10	&& 0.73 & 0.73 & .06 & .07	&& 0.11 & 0.11 & .01 & .01	&& 3.60 & 3.68 & .07 & .08\\
& $0.4$ && 1.26 & 1.27 & .10 & .10	&& 0.99 & 0.98 & .07 & .07	&& 0.24 & 0.23 & .02 & .02	&& 4.20 & 4.23 & .08 & .08\\
& $0.6$ && 1.46 & 1.47 & .10 & .10	&& 1.01 & 1.00 & .07 & .07	&& 0.41 & 0.40 & .03 & .03	&& 4.80 & 4.78 & .10 & .10\\
& $0.8$ && 1.80 & 1.80 & .10 & .10	&& 1.27 & 1.25 & .08 & .08	&& 0.66 & 0.64 & .05 & .06	&& 5.40 & 5.33 & .12 & .13\\
& $v$ && 
	$\gamma_0$ & $\hat\gamma_0$ & $\text{se}$ & $\hat{\text{se}}$ &&
	$\gamma_1$ & $\hat\gamma_1$ & $\text{se}$ & $\hat{\text{se}}$ &&
	$\gamma_0$ & $\hat\gamma_0$ & $\text{se}$ & $\hat{\text{se}}$ &&
	$\gamma_1$ & $\hat\gamma_1$ & $\text{se}$ & $\hat{\text{se}}$\\
\cline{4-7}\cline{9-12}\cline{14-17}\cline{19-22}
& $0.2$ && 0.42 & 0.40 & .09 & .09	&& 0.42 & 0.41 & .08 & .08	&& 0.20 & 0.21 & .03 & .04	&& 0.10 & 0.10 & .02 & .03\\
& $0.4$ && 0.13 & 0.12 & .03 & .03	&& 0.13 & 0.12 & .02 & .02	&& 0.41 & 0.42 & .07 & .07	&& 0.21 & 0.20 & .05 & .05\\
& $0.6$ && 0.13 & 0.12 & .03 & .03	&& 0.13 & 0.12 & .02 & .02	&& 0.65 & 0.66 & .11 & .11	&& 0.33 & 0.32 & .08 & .08\\
& $0.8$ && 0.42 & 0.40 & .09 & .09	&& 0.42 & 0.41 & .08 & .08	&& 0.96 & 0.98 & .16 & .17	&& 0.48 & 0.47 & .12 & .12\\
\hline
\hline
\cline{4-12}\cline{14-22}
&&& \multicolumn{4}{l}{$T = 5$} && \multicolumn{4}{l}{$T = 10$} && \multicolumn{4}{l}{$T = 5$} && \multicolumn{4}{l}{$T = 10$}\\
\cline{4-7}\cline{9-12}\cline{14-17}\cline{19-22}
& $\alpha$ &&  
	$\tilde{\alpha}$& &   \multicolumn{2}{l}{$1 - \tilde{\beta}$} &&
	$\tilde{\alpha}$& &   \multicolumn{2}{l}{$1 - \tilde{\beta}$} &&
	$\tilde{\alpha}$& &   \multicolumn{2}{l}{$1 - \tilde{\beta}$} &&
	$\tilde{\alpha}$& &   \multicolumn{2}{l}{$1 - \tilde{\beta}$}\\
& $0.05$ && 0.05 && 0.22 &&& 0.04 && 0.50 &&& 0.04 && 0.30 &&& 0.04 && 0.53\\
& $0.10$ && 0.10 && 0.36 &&& 0.09 && 0.80 &&& 0.10 && 0.45 &&& 0.09 && 0.69\\
\hline
\hline
\end{tabular}
\footnotesize\center
Summary of simulation results, based on $R = 1000$ Monte Carlo replications, with $N = 150$ and $T = \{5, 10\}$. 
For each coefficient, we report the true \textit{absolute} value ($\beta, \gamma$) at the quintiles $(0.2, 0.4, 0.6, 0.8)$,
the average estimate ($\hat\beta, \hat\gamma$), the standard error ($\text{se}$), computed as the standard deviation
of the estimated model parameters across simulations, and the average estimated asymptotic standard error ($\hat{\text{se}}$). 
The bottom table reports, for two different nominal levels $\alpha = 0.05, 0.10$, the empirical probability
of type I error ($\tilde{\alpha}$) and the power ($1 - \tilde{\beta}$) of the Kolmogorov-Smirnov goodness-of-fit test described in Section \ref{sec:gof}.
\end{sidewaystable}


\begin{sidewaystable}
\footnotesize
\caption{Simulation results with $N = 300$}
\label{sim300}
\centering
\begin{tabular}{cc
m{0.01cm}m{0.45cm}m{0.45cm}m{0.2cm}m{0.2cm}
m{0.01cm}m{0.45cm}m{0.45cm}m{0.2cm}m{0.2cm}
m{0.01cm}m{0.45cm}m{0.45cm}m{0.2cm}m{0.2cm}
m{0.01cm}m{0.45cm}m{0.45cm}m{0.2cm}m{0.2cm}
}
\hline
\hline
&&& \multicolumn{9}{c}{Simulation 1} & & \multicolumn{9}{c}{Simulation 2}\\
\cline{4-12}\cline{14-22}
$T = 5$ & $u$ && 
	$\beta_0$ & $\hat\beta_0$ & $\text{se}$ & $\hat{\text{se}}$ &&
	$\beta_1$ & $\hat\beta_1$ & $\text{se}$ & $\hat{\text{se}}$ &&
	$\beta_0$ & $\hat\beta_0$ & $\text{se}$ & $\hat{\text{se}}$ &&
	$\beta_1$ & $\hat\beta_1$ & $\text{se}$ & $\hat{\text{se}}$\\
\cline{4-7}\cline{9-12}\cline{14-17}\cline{19-22}
& $0.2$ && 1.11 & 1.17 & .08 & .08	&& 0.73 & 0.72 & .07 & .07	&& 0.11 & 0.11 & .01 & .01	&& 3.60 & 3.79 & .08 & .08\\
& $0.4$ && 1.26 & 1.30 & .08 & .08	&& 0.99 & 0.97 & .08 & .08	&& 0.24 & 0.24 & .02 & .02	&& 4.20 & 4.28 & .08 & .08\\
& $0.6$ && 1.46 & 1.48 & .08 & .08	&& 1.01 & 1.00 & .08 & .08	&& 0.41 & 0.41 & .03 & .04	&& 4.80 & 4.77 & .10 & .10\\
& $0.8$ && 1.80 & 1.79 & .08 & .08	&& 1.27 & 1.24 & .09 & .09	&& 0.66 & 0.66 & .06 & .06	&& 5.40 & 5.27 & .13 & .13\\
& $v$ && 
	$\gamma_0$ & $\hat\gamma_0$ & $\text{se}$ & $\hat{\text{se}}$ &&
	$\gamma_1$ & $\hat\gamma_1$ & $\text{se}$ & $\hat{\text{se}}$ &&
	$\gamma_0$ & $\hat\gamma_0$ & $\text{se}$ & $\hat{\text{se}}$ &&
	$\gamma_1$ & $\hat\gamma_1$ & $\text{se}$ & $\hat{\text{se}}$\\
\cline{4-7}\cline{9-12}\cline{14-17}\cline{19-22}
& $0.2$ && 0.42 & 0.37 & .06 & .07	&& 0.42 & 0.41 & .06 & .06	&& 0.20 & 0.20 & .03 & .03	&& 0.10 & 0.10 & .02 & .02\\
& $0.4$ && 0.13 & 0.11 & .02 & .02	&& 0.13 & 0.12 & .02 & .02	&& 0.41 & 0.41 & .05 & .06	&& 0.21 & 0.20 & .04 & .04\\
& $0.6$ && 0.13 & 0.11 & .02 & .02	&& 0.13 & 0.12 & .02 & .02	&& 0.65 & 0.65 & .09 & .09	&& 0.33 & 0.32 & .06 & .06\\
& $0.8$ && 0.42 & 0.37 & .06 & .07	&& 0.42 & 0.41 & .06 & .06	&& 0.96 & 0.96 & .13 & .13	&& 0.48 & 0.47 & .09 & .09\\
$T = 10$ & $u$ && 
	$\beta_0$ & $\hat\beta_0$ & $\text{se}$ & $\hat{\text{se}}$ &&
	$\beta_1$ & $\hat\beta_1$ & $\text{se}$ & $\hat{\text{se}}$ &&
	$\beta_0$ & $\hat\beta_0$ & $\text{se}$ & $\hat{\text{se}}$ &&
	$\beta_1$ & $\hat\beta_1$ & $\text{se}$ & $\hat{\text{se}}$\\
\cline{4-7}\cline{9-12}\cline{14-17}\cline{19-22}
& $0.2$ && 1.11 & 1.14 & .07 & .07	&& 0.73 & 0.73 & .05 & .05	&& 0.11 & 0.10 & .01 & .01	&& 3.60 & 3.67 & .05 & .05\\
& $0.4$ && 1.26 & 1.27 & .07 & .07	&& 0.99 & 0.98 & .05 & .05	&& 0.24 & 0.23 & .01 & .01	&& 4.20 & 4.23 & .06 & .06\\
& $0.6$ && 1.46 & 1.47 & .07 & .07	&& 1.01 & 1.00 & .05 & .05	&& 0.41 & 0.40 & .02 & .02	&& 4.80 & 4.78 & .07 & .07\\
& $0.8$ && 1.80 & 1.80 & .07 & .07	&& 1.27 & 1.25 & .05 & .06	&& 0.66 & 0.64 & .04 & .04	&& 5.40 & 5.33 & .09 & .09\\
& $v$ && 
	$\gamma_0$ & $\hat\gamma_0$ & $\text{se}$ & $\hat{\text{se}}$ &&
	$\gamma_1$ & $\hat\gamma_1$ & $\text{se}$ & $\hat{\text{se}}$ &&
	$\gamma_0$ & $\hat\gamma_0$ & $\text{se}$ & $\hat{\text{se}}$ &&
	$\gamma_1$ & $\hat\gamma_1$ & $\text{se}$ & $\hat{\text{se}}$\\
\cline{4-7}\cline{9-12}\cline{14-17}\cline{19-22}
& $0.2$ && 0.42 & 0.39 & .06 & .07	&& 0.42 & 0.41 & .05 & .06	&& 0.20 & 0.21 & .02 & .02	&& 0.10 & 0.10 & .02 & .02\\
& $0.4$ && 0.13 & 0.12 & .02 & .02	&& 0.13 & 0.12 & .02 & .02	&& 0.41 & 0.42 & .05 & .05	&& 0.21 & 0.20 & .03 & .04\\
& $0.6$ && 0.13 & 0.12 & .02 & .02	&& 0.42 & 0.12 & .02 & .02	&& 0.65 & 0.66 & .07 & .08	&& 0.33 & 0.32 & .05 & .06\\
& $0.8$ && 0.42 & 0.39 & .06 & .07	&& 0.13 & 0.41 & .05 & .06	&& 0.96 & 0.98 & .11 & .12	&& 0.48 & 0.48 & .08 & .08\\
\hline
\hline
\cline{4-12}\cline{14-22}
&&& \multicolumn{4}{l}{$T = 5$} && \multicolumn{4}{l}{$T = 10$} && \multicolumn{4}{l}{$T = 5$} && \multicolumn{4}{l}{$T = 10$}\\
\cline{4-7}\cline{9-12}\cline{14-17}\cline{19-22}
& $\alpha$ &&  
	$\tilde{\alpha}$& &   \multicolumn{2}{l}{$1 - \tilde{\beta}$} &&
	$\tilde{\alpha}$& &   \multicolumn{2}{l}{$1 - \tilde{\beta}$} &&
	$\tilde{\alpha}$& &   \multicolumn{2}{l}{$1 - \tilde{\beta}$} &&
	$\tilde{\alpha}$& &   \multicolumn{2}{l}{$1 - \tilde{\beta}$}\\
& $0.05$ && 0.07 && 0.43 &&& 0.05 && 0.99 &&& 0.04 && 0.52 &&& 0.06 && 0.81\\
& $0.10$ && 0.11 && 0.65 &&& 0.10 && 1.00 &&& 0.08 && 0.67 &&& 0.10 && 0.90\\
\hline
\hline
\end{tabular}
\footnotesize\center
Summary of simulation results with $N = 300$ and $T = \{5, 10\}$. 
\end{sidewaystable}

\clearpage
\section{ANALYSIS OF NGAL DATA}\label{sec:app}

We analyzed data from \cite{app}, aiming to investigate the role of plasma neutrophil gelatinase-associated lipocalin (NGAL)
as a marker of sepsis and acute kidney disfunction. The dataset included 139 patients admitted to the general intensive care unit at
Karolinska University Hospital in Solna, Sweden, between August 2007 and November 2010. Baseline information was collected, and patients 
were classified daily as having sepsis or not. NGAL (mg/mL), procalcitonin (PCT), C-reactive protein (CRP), and 
creatinine changes relative to baseline ($\Delta{\text{creat}}$) were measured daily before discharge, for a total of 1317 plasma samples. 
After removing missing data, individuals with only one observation, and one patient with severe complications, the final sample
included 135 patients for a total sample size of $\sum_{i = 1}^{135} T_i = 1263$. The number of observations per patient
varied between $2$ and $38$, and more than $80\%$ of patients had $T_i \le 14$. 

The goal of our analysis was to estimate conditional quantiles of NGAL, and in particular to measure its association with sepsis.
The between-patient variability appeared to be very large, reflecting the presence of important differences in the initial health conditions. 
We formulated a regression model with the following predictors: a binary indicator of sepsis status, an indicator of $\Delta{\text{creat}} \ge 50$, 
age (centered at its median, 52 years, and divided by 10), an indicator of female gender, and time since hospitalization (weeks). Age and gender
were cluster-invariant and were also included as level-2 predictors. 

The response variable was log-transformed, which made it more plausible to define individual
effects on the additive scale as in model (\ref{themodel}). 
The regression function was
\begin{eqnarray}
\log(\text{NGAL}_{it}) & = & \beta_0(U_{it}) + \beta_1(U_{it})\text{sepsis}_{it} + \beta_2(U_{it})I(\Delta{\text{creat}}_{it} > 50)  \nonumber \\
 & + & \beta_3(U_{it})(\text{age}_{i} - 52)/10 + \beta_4(U_{it})\text{female}_{i} + \beta_5(U_{it})\text{time}_{it}  \nonumber \\
 & + & \gamma_0(V_{i}) + \gamma_1(V_{i})(\text{age}_{i} - 52)/10 + \gamma_2(V_{i})\text{female}_{i}. \nonumber
\end{eqnarray}
We formulated a variety of models, in which $\beta_0(\cdot)$ and $\gamma_0(\cdot)$ were  unbounded,  
while the other coefficients were modeled by bounded functions.
To facilitate interpretation, we forced $\gammavec(0.5) = 0$, assigning the
individual effects a zero-median distribution in which level-2 covariates only affect the scale parameter.

Selected modeling options are illustrated in Table \ref{spec}. 
Different models appeared to fit the data well, and were not rejected by the goodness-of-fit
test described in Section \ref{sec:gof}. The following model combined simplicity
and flexibility, and was selected for illustrative purposes:
\begin{eqnarray}
\beta_0(u \mid \thetavec) &=& \theta_{00} + \theta_{01}\log(u) + \theta_{02}\log(1 - u), \nonumber\\
\beta_j(u \mid \thetavec) &=& \theta_{j0} + \theta_{j3}u + \theta_{j4}u^{1/4} + \theta_{j5}(1 - u)^{1/4}, \hspace{0.5cm} j = 1, \ldots, 5, \nonumber\\
\gamma_0(v \mid \phivec) &=& \phi_{01}\log(2v) + \phi_{02}\log(2(1 - v)), \nonumber\\
\gamma_j(v \mid \phivec) &=& \phi_{j3}(v - 0.5), \hspace{0.5cm} j = 1, \ldots, 2. \nonumber
\end{eqnarray}
The level-1 and level-2 intercepts were described by different versions of the asymmetric Logistic distribution.
The coefficient functions associated with level-1 covariates were a combination of linear and root-4 functions,
while those of the level-2 predictors were assumed to be linear.

The p-value of the Kolmogorov-Smirnov test was $0.21$. To assess local fit, the test was repeated 
in subsamples with different values of the covariates (e.g., the females, those with $\Delta{\text{creat}} > 50$, etc.).
No significant evidence of model misspecification was found.

All 27 model parameters are reported in Table \ref{est_csi}, while regression coefficients at selected quantiles are
summarized in Table \ref{est_beta}. We represent graphically the quantile regression coefficient functions in Figure 2,
where we also report a ``nonparametric'' fit obtained by modeling all coefficients 
as piecewise linear functions with knots at the deciles.

Results showed that the distribution of the individual effects was almost symmetric (as suggested by the fact that
$\hat\phi_{01} \simeq -\hat\phi_{02}$) and that its variance was not significantly affected by 
cluster-level predictors. Instead, all predictors apart from gender appeared to be associated with the level-1 response. 
In particular, the coefficients associated with sepsis, $\Delta{\text{creat}}_{it} > 50$ and age were consistently positive
at all quantiles, while the coefficient of time was always negative. 
The sepsis status was associated with a percentile difference of about 0.12 at quantiles 0.2, 0.4, 0.6, 0.8. 
As shown by Figure 2, an even larger percentile difference was found at quantiles above 0.8.


\begin{table}[h]
\caption{Alternative specifications of $\bvec(u)$ and $\cvec(v)$}
\label{spec}
\small
\centering
\begin{tabular}{cccc}
\noalign{\vspace{0.2cm}}
\hline
\hline
\noalign{\vspace{0.2cm}}
$\beta_0(u)$ & $\beta_j(u)$ & $\gamma_0(v)$ & $\gamma_j(v)$ \\
\noalign{\vspace{0.1cm}}
\cmidrule(lr){1-1}\cmidrule(lr){2-2}\cmidrule(lr){3-3}\cmidrule(lr){4-4}
$\zeta(u)$ & $u$ & $\zeta(v)$ & $(v - 0.5)$\\
$\log(\frac{u}{1 - u})$ & $u, u^2, u^3, \ldots$ & $\log(\frac{v}{1 - v})$ & $(v - 0.5)^3$\\
$\log(u), \log(1 - u)$ & $u, \cos(\pi u), \sin(\pi u)$ & $\log(2v), \log(2(1 - v))$ & $(v - 0.5), (v - 0.5)^3$\\
$u, \log(u), \log(1 - u)$ &  $u^{\delta}, (1 - u)^{\delta}$ & $v - 0.5, \log(2v), \log(2(1 - v))$ & $(v - 0.5)^3_{-}, (v - 0.5)^3_{+}$\\
\noalign{\vspace{0.1cm}}
\hline
\hline
\noalign{\vspace{0.1cm}}
\end{tabular}
\vspace{0.1cm}\footnotesize\center
In the table, $\beta_0(u)$ and $\gamma_0(v)$ denote level-1 and level-2 intercept, respectively,
while $\beta_j(u)$ and $\gamma_j(v)$ represent coefficients associated with generic level-1 and level-2
covariates. Different parametric models are represented by the functions that compose $\bvec(u)$ 
and $\cvec(v)$. A constant term $b(u) = 1$ was always included. 
The notation $\zeta(\cdot)$ is used for the quantile function of a standard normal distribution,
and we defined $(v - 0.5)^3_{-} = I(v \le 0.5)(v - 0.5)^3$ and $(v - 0.5)^3_{+} = I(v > 0.5)(v - 0.5)^3$.
\end{table}


\begin{figure}[h]
    \centering    
	\makebox{\includegraphics[scale = 0.5]{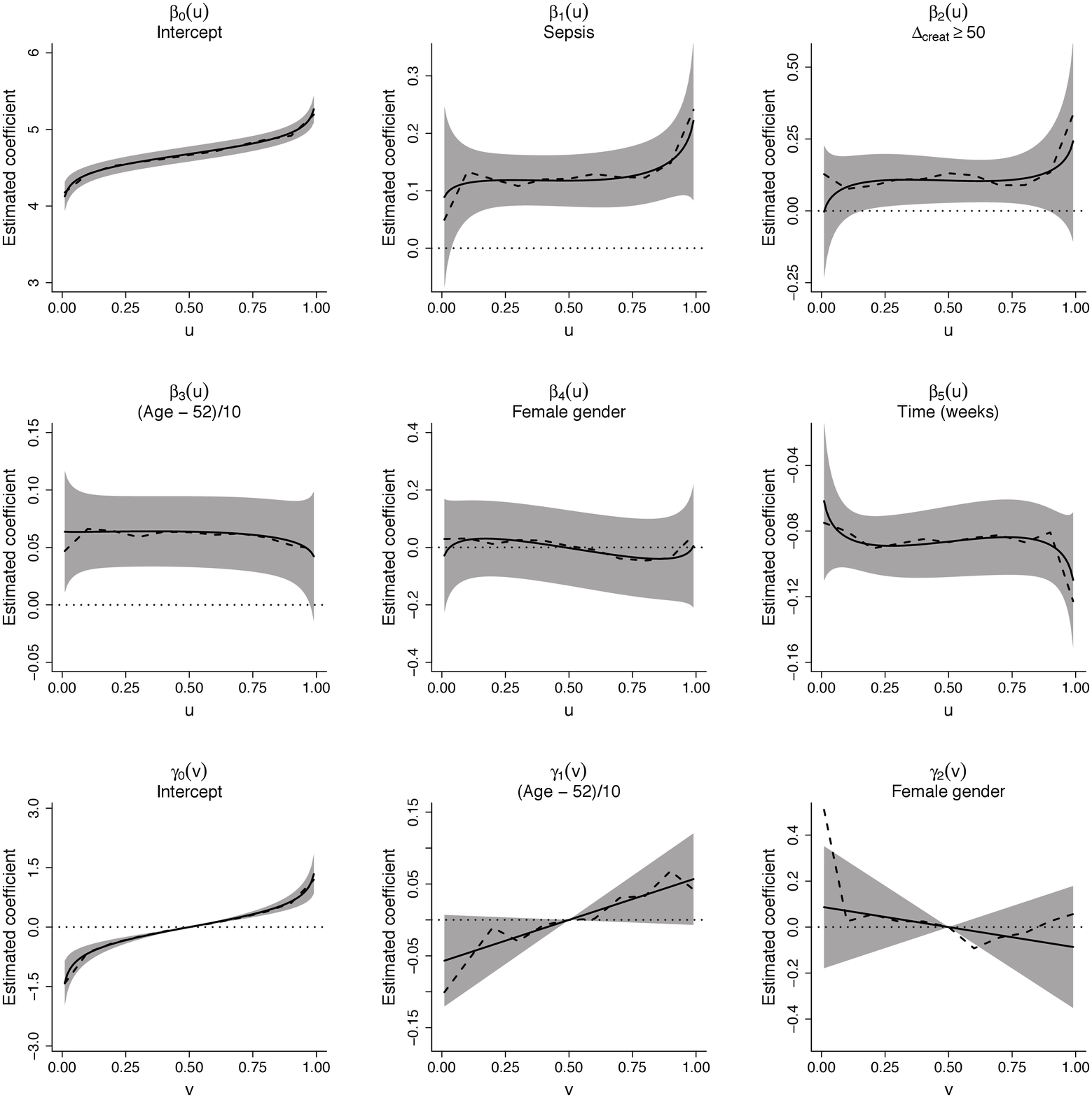}}
	\caption{\label{fig2}
	Continuous lines represent the estimated quantile regression coefficient functions, based on the parametric model
	summarized in Table \ref{est_csi}. Shaded areas represents pointwise 95$\%$ confidence intervals.
	The dashed lines are obtained from a ``nonparametric'' model in which
 	$\betavec(u)$ and $\gammavec(v)$ were fitted by piecewise linear functions with knots at the deciles.
	The dotted horizontal lines indicate the zero.}
\end{figure}


\begin{table}[h]
\caption{Estimated model parameters}
\label{est_csi}
\small
\centering
\begin{tabular}{ccccccc}
\noalign{\vspace{0.2cm}}
Level 1 ($\hat\thetavec$) & 1 & $\log(u)$ & $\log(1 - u)$ & $u$ & $u^{1/4}$ & $(1 - u)^{1/4}$\\
\noalign{\vspace{0.1cm}}
\hline
\noalign{\vspace{0.1cm}}
Intercept 						& 4.67 (.07) & 0.12 (.03) & -0.13 (.03) & - & - & - 			\\
Sepsis 							& 0.39 (.26) & - & - & -0.20 (.23) & 0.14 (.29) & -0.35 (.27) 	\\
$\Delta{\text{creat}} \ge 50$ 	& 0.39 (.61) & - & - & -0.44 (.40) & 0.46 (.46) & -0.54 (.57) 	\\
(age - 52)/10 					& 0.00 (.09) & - & - & 0.02 (.07) & 0.00 (.08) & 0.06 (.09) 	\\
female gender 				& 0.16 (.31) & - & - & -0.40 (.23) & 0.34 (.28) & -0.29 (.30) 	\\
time (weeks) 					& -0.14 (.09) & - & - & 0.12 (.08) & -0.12 (.10) & 0.12 (.09) 	\\
\noalign{\vspace{0.1cm}}
\hline
\hline
\noalign{\vspace{0.1cm}}
\end{tabular}
\begin{tabular}{cccc}
\noalign{\vspace{0.2cm}}
 Level 2 ($\hat\phivec$) & $\log(2v)$ & $\log(2(1 - v))$ & $(v - 0.5)$\\
\noalign{\vspace{0.1cm}}
\hline
\noalign{\vspace{0.1cm}}
Intercept 			& 0.33 (.08) & -0.29 (.08) & - 	\\
(age - 52)/10		& - & - & 0.11 (.07) 			\\
female gender 	& - & - & -0.20 (.28) 			\\
\noalign{\vspace{0.1cm}}
\hline
\hline
\noalign{\vspace{0.1cm}}
\end{tabular}
\vspace{0.1cm}\footnotesize\center
Summary of the selected model (top: $\hat \thetavec$; bottom: $\hat\phivec$),
with estimated standard errors in brackets. 
The model is represented graphically in Figure 2, 
and selected quantiles are summarized in Table \ref{est_beta}.
\end{table}


\begin{table}[h]
\caption{Summary of selected quantiles}
\label{est_beta}
\small
\centering
\begin{tabular}{cccccc|c}
\noalign{\vspace{0.1cm}}
\hline
\hline
\noalign{\vspace{0.2cm}}
quantile & 0.2 & 0.4 & 0.6 & 0.8\\
\noalign{\vspace{0.1cm}}
\hline
\noalign{\vspace{0.1cm}}
Intercept & 4.51 (.05) & 4.63 (.05) & 4.73 (.05) & 4.85 (.05)\\
Sepsis & 0.12 (.02) & 0.12 (.02) & 0.12 (.02) & 0.13 (.03)\\
$\Delta{\text{creat}}$ & 0.10 (.05) & 0.11 (.04) & 0.10 (.04) & 0.11 (.05)\\
(age - 52)/10 & 0.06 (.02) & 0.06 (.02) & 0.06 (.02) & 0.06 (.02)\\
female gender & 0.03 (.07) & 0.01 (.07) & -0.02 (.07) & -0.04 (.07)\\
time (weeks) & -0.09 (.01) & -0.09 (.01) & -0.08 (.01) & -0.08 (.01)\\
\noalign{\vspace{0.1cm}}
\hline
\noalign{\vspace{0.1cm}}
Intercept & -0.42 (.06) & -0.12 (.01) & 0.12 (.01) & 0.41 (.05)\\
(age - 52)/10 & -0.03 (.02) & -0.01 (.01) & 0.01 (.01) & 0.03 (.02)\\
female gender & 0.05 (.08) & 0.02 (.03) & -0.02 (.03) & -0.05 (.08)\\
\noalign{\vspace{0.1cm}}
\hline
\hline
\noalign{\vspace{0.1cm}}
\end{tabular}
\vspace{0.1cm}\footnotesize\center
Estimated regression coefficients at quantiles $(0.2, 0.4, 0.6, 0.8)$.
Top table: $\hat\betavec(u) = \betavec(u \mid \hat\thetavec) = \hat\thetavec \bvec(u)$;
bottom table: $\hat\gammavec(v) = \gammavec(v \mid \hat\phivec) = \hat\phivec \cvec(v)$.
Estimated standard errors in brackets.
\end{table}


\clearpage
\section{CONCLUSIONS}

We introduced a general framework for longitudinal quantile regression, extending the work of \citeauthor{iqr} (\citeyear{iqr, ctiqr})
on quantile regression coefficients modeling. 
We defined a two-level quantile function in which both the ``within'' and the ``between'' part of the distribution are 
described by a quantile regression model. This allows to investigate how covariates affect not only the level-1 response,
but also the distribution of the individual effects, which is generally overlooked in the existing literature on longitudinal 
quantile regression.
Identification is achieved by modeling the coefficient functions parametrically, and estimation is carried out
by minimizing a smooth objective function.

The proposed method is computationally simple and can be viewed as a 
special type of penalized fixed-effects estimator that presents important elements of novelty of its own.
The penalty term carries information on the parameters that describe the conditional distribution
of the individual effects. This permits estimating both level-1 and level-2 parameters, as in random-effects models,
but allows carrying out estimation and inference using fixed-effects techniques. Moreover, it avoids the problem of choosing
a tuning constant as in standard $\ell_1$- or $\ell_2$-penalization. The described form
of penalized fixed-effects method is not limited to a quantile regression framework and could be applied
to different estimation problems.

The proposed modeling framework can be generalized in different directions. 
An interesting possibility is to include multiple individual effects as in random-slope models.
In our framework, individual effects are represented by a pure location shift as in \cite{koenker2004},
\citeauthor{geraci} (\citeyear{geraci}, \citeyear{geraci2}), and \cite{canay}. Using the proposed
penalized fixed-effects approach, it is relatively simple to incorporate
not only an individual intercept, $\alpha_i$, but also a set of individual slopes, say $\{\delta_{1i}, \delta_{2i}, \ldots\}$.
This, however, would typically result in cumbersome computation and, unless
$N$ and $T$ are sufficiently large, would probably undermine model identifiability.
Using \cites{koenker2004} words: ``At best we may be able to estimate an individual specific location-shift effect, 
and even this may strain credulity''.

Another interesting extension is represented by varying-coefficients models (e.g., \citealp{hastie, fan1, fan2, chiang,kim2})
that could be implemented by allowing the level-1 regression coefficients to be functions of time. A possible approach is 
to describe the coefficients, say $\betavec(u, t)$, using tensor products
of splines. Finally, the proposed method could be used to estimate
static and dynamic quantile autoregressive models (e.g., \citealp{arellano}).

An important problem that has not been discussed in the paper is represented by quantile crossing, occurring
when either $\xx_{it}^\T\betavec^{\prime}(u \mid \hat\thetavec) < 0$ or $\zz_{i}^\T\gammavec^{\prime}(v \mid \hat\phivec) < 0$. 
One may want to determine in advance which values of the parameters $\thetavec$ and $\phivec$ 
would ensure that no crossing occurs, i.e., that $\xx_{it}^\T\betavec(\cdot \mid \hat\thetavec)$ and $\zz_{i}^\T\gammavec(\cdot \mid \hat\phivec)$
are monotonically increasing functions. This is only possible in very simple models with few covariates, or in presence of
restrictive assumptions. However, simulation evidence suggests that parametric models are relatively immune to quantile crossing,
compared with the ``nonparametric'' approaches based on ordinary quantile regression.
Additionally, the parametric structure makes it particularly simple to verify crossing, taking advantage
of the closed-form analytical expression of the quantile function, and admits the application of monotonization methods such as the rearrangement of \cite{cfg10} to produce increasing estimates of conditional quantiles.

This paper is accompanied by an R package \texttt{qrcm}, that includes a function named \texttt{iqrL}
that performs model fitting, and a variety of auxiliary functions for prediction, plotting, and goodness-of-fit assessment.
The documentation contains a rich set of examples, and can serve as tutorial for the practitioners.
The package is available upon request to the authors.

%
%
%
%
%
%
%
%
%
%
%
%
%
%
%

\section*{Appendix A - Proof of Theorem 1} 
\setcounter{theorem}{1}
\setcounter{assumption}{1}

\subsection*{Asymptotic Expansion for Estimators with Different Rates of Convergence}

We start by extending  the asymptotic expansions of Hahn and Newey (2004) and Fern\'andez-Val  (2005) for fixed effects estimators in nonlinear panel data models to models with mixed-rates asymptotics. 
In particular, we consider the  M-estimators:
\begin{equation} \label{eq:est}
(\hat \thetavec, \hat \phivec) \in \argmin_{\{\thetavec \in \Theta,\phivec \in \Phi\}} L(\thetavec,\phivec), \ \  \hat{L}(\thetavec,\phivec) = \frac{1}{NT} \left[\sum_{i=1}^N \sum_{t=1}^T \ell_{1it}(\thetavec,\hat \alpha_i(\thetavec,\phivec)) +   \sum_{i=1}^N\ell_{2i}(\phivec,\hat \alpha_i(\thetavec,\phivec))\right], 
\end{equation}
where 
\begin{equation}\label{eq:est2}
\hat \alpha_i(\thetavec,\phivec) \in  \argmin_{\alpha_i \in \mathcal{A}}  \hat{L}_{i}(\thetavec,\phivec,\alpha_i), \ \ \hat{L}_{i}(\thetavec,\phivec,\alpha_i) = \frac{1}{T} \left[\sum_{t=1}^T \ell_{1it}(\thetavec, \alpha_i) +   \ell_{2i}(\phivec, \alpha_i)\right],
\end{equation}
$\ell_{1it}$ and $\ell_{2i}$ are random functions that depend on the data, and $\mathcal{A}$, $\Theta$ and $\Phi$ denote the parameter spaces for $\alpha_i$, $\thetavec$ and $\phivec$, respectively.

Let  $\mathcal{B}_{\varepsilon}(\alpha_{i}^0)$, $1 \leq i \leq n$, $\mathcal{B}_{\varepsilon}(\thetavec^0)$ and $\mathcal{B}_{\varepsilon}(\phivec^0)$ be $\varepsilon$-neighborhoods of the true value of the parameters $\alpha_{i}^0$, $1 \leq i \leq n$, $\thetavec^0$ and $\phivec^0$, for some $\varepsilon >0$. In what follows, we  use superscripts for partial derivatives, e.g.  $\ell_{1it}^{\thetavec}(\thetavec, \alpha_i) := \partial \ell_{1it}(\thetavec,\alpha_i)/\partial \thetavec$ and $\ell_{1it}^{\alpha \alpha}(\thetavec,\alpha_i) := \partial^2 \ell_{1it}(\thetavec,\alpha_i)/\partial \alpha_i^2$, and often drop the arguments when the functions are evaluated at the true values, e.g. $\ell_{1it}^{\thetavec} := \ell_{1it}^{\thetavec}(\thetavec^0, \alpha_i^0)$. We assume that the parameters $\thetavec$ and $\phivec$ are vector-valued. We explain below how to adapt the expansions to the integrated loss minimization estimator where $\thetavec$ and $\phivec$ can be matrices.

\begin{assumption}\label{ass:gral} 
(i) Consistency: $\hat \thetavec \to_P \thetavec^0$, $\hat \phivec \to_P \phivec^0$ and $\sup_{i}|\hat \alpha_i - \alpha_i^0| \to_P 0 $ for $\hat \alpha_i := \hat \alpha_i(\hat \thetavec, \hat \phivec)$. 
(ii) For each $i$, $(\thetavec, \alpha_i) \mapsto \ell_{1it}(\thetavec, \alpha_i)$ is four times continuously differentiable a.s. on $\mathcal{B}_{\varepsilon}(\thetavec^0) \times \mathcal{B}_{\varepsilon}(\alpha_{i}^0)$ and the partial derivatives up to fourth order are  bounded in absolute value by random variables $M_{it} >0$ a.s. such that $\sup_{i} \plim_{T\to \infty} \sum_{t=1}^T M_{it}^{\xi} / T < \infty$  for some $\xi > 8$; and  $(\phivec, \alpha_i) \mapsto  \ell_{2i}(\phivec, \alpha_i)$ is four times continuously differentiable a.s. on  $\mathcal{B}_{\varepsilon}(\phivec^0)\times \mathcal{B}_{\varepsilon}(\alpha_{i}^0)$  and the partial derivatives up to fourth order are  bounded in absolute value a.s. (iii) The true value of the parameters are in the interiors of the parameter spaces, i.e., $\thetavec^0 \in \Theta^\circ$, $\phivec^0 \in \Phi^\circ$, and $\alpha_i^0 \in \mathcal{A}^\circ$ for each $i$. (iv) The following limits exist:
\begin{eqnarray*}
\bar{\H}_{\thetavec} &=&  \plim_{N,T \to \infty} \frac{1}{NT} \sum_{i=1}^N \sum_{t=1}^T \left[ \ell_{1it}^{\thetavec \thetavec} - \ell_{1it}^{\thetavec \alpha} \frac{\sum_{t=1}^T \ell_{1it}^{\alpha \thetavec } }{\sum_{t=1}^T  \ell_{1it}^{\alpha \alpha }} \right], \\
\bar{\H}_{\phivec\thetavec} &=&  - \plim_{N,T \to \infty} \frac{1}{N} \sum_{i=1}^N  \ell_{2i}^{\phivec \alpha} \frac{\sum_{t=1}^T \ell_{1it}^{\alpha \thetavec } }{\sum_{t=1}^T  \ell_{1it}^{\alpha \alpha }}, \ \
\bar{\H}_{\phivec} =  \plim_{N \to \infty} \frac{1}{N} \sum_{i=1}^N  \ell_{2i}^{\phivec \phivec}, \\
\bar{b}_{\thetavec} &=&   \plim_{N,T \to \infty} \frac{1}{N} \sum_{i=1}^N b_{\thetavec i}, \ \
\bar{b}_{\phivec} =   \plim_{N,T \to \infty} \frac{1}{N} \sum_{i=1}^N b_{\phivec i}, \\
\bar{\O}_{\thetavec} &=& \plim_{N,T \to \infty} \frac{1}{N} \sum_{i=1}^N \varphi_{\thetavec i} \varphi_{\thetavec i}^{\T}, \ \ \bar{\O}_{\phivec} = \plim_{N,T \to \infty} \frac{1}{N} \sum_{i=1}^N \varphi_{\phivec i} \varphi_{\phivec i}^{\T}, 
\end{eqnarray*}
where 
\begin{eqnarray*}
b_{\thetavec i} &=&  \frac{1}{T} \sum_{t=1}^T  \ell_{1it}^{\thetavec \alpha} \beta_i -   \frac{\sigma_i^2}{T} \sum_{t=1}^T  \ell_{1it}^{\thetavec \alpha} \ell_{1it}^{\alpha}  + \frac{1}{2T} \sum_{t=1}^T \ell_{1it}^{\thetavec \alpha \alpha} \psi_i^2, \\
b_{\phivec i} &=&    \ell_{2i}^{\phivec \alpha}(\sqrt{T}    \psi_i  + \beta_i)  + \frac{1}{2}  \ell_{2i}^{\phivec \alpha \alpha}  \psi_i^2, 
 \\
\varphi_{\thetavec i} &=& \frac{1}{\sqrt{T}} \sum_{t=1}^T \ell_{1it}^{\thetavec} +   \frac{\psi_i}{T} \sum_{t=1}^T \ell_{1it}^{\thetavec \alpha}, \ \ \varphi_{\phivec i} =  \ell_{2i}^{\phivec},
\\
\beta_i &=& \sigma_i^2 \plim_{T\to \infty} \frac{1}{T} \sum_{t=1}^T \left(\sigma_i^2 \ell_{1it}^{\alpha \alpha} \ell_{1it}^{\alpha} - \frac{\psi_i^2}{2} \ell_{1it}^{\alpha \alpha \alpha}   \right),\\
\psi_i &=& - \frac{\sigma_i^2}{\sqrt{T}} \left(  \sum_{t=1}^T \ell_{1it}^{\alpha} + \ell_{2i}^{\alpha} \right),  \ \ \sigma_i^2 =  \left[ \plim_{T\to \infty} \frac{1}{T} \sum_{t=1}^T \ell_{1it}^{\alpha\alpha} \right]^{-1}.
\end{eqnarray*}
(v) The minimum eigenvalues of the matrices $\bar{\H}_{\thetavec}$ and $\bar{\H}_{\phivec}$ are bounded away from zero and $\inf_i \plim_{T \to \infty} \sum_{t=1}^T  \ell_{1it}^{\alpha\alpha}/T$ is bounded away from zero.

\end{assumption}

\begin{theorem}\label{thm:gral} Suppose that Assumption \ref{ass:gral} holds. If $N = O(T)$,
$$
\sqrt{NT}\left( \hat \thetavec - \thetavec^0 + \frac{\bar{\H}_{\thetavec}^{-1} \bar{b}_{\thetavec}}{T} \right) \to_d \bar{\H}_{\thetavec}^{-1} \mathcal{N}(0, \bar{\O}_{\thetavec}). 
$$ 
If $N = O(T^2)$,
$$
\sqrt{N}\left( \hat \phivec - \phivec^0 + \frac{\bar{\H}_{\phivec}^{-1} (\bar{b}_{\phivec} +  \bar{\H}_{\phivec\thetavec} \bar{\H}_{\thetavec}^{-1}\bar{b}_{\thetavec})}{T} \right) \to_d \bar{\H}_{\phivec}^{-1} \mathcal{N}(0, \bar{\O}_{\phivec}). 
$$ 
\end{theorem}

\begin{proof}
We divide the proof in two steps: asymptotic expansions and asymptotic distributions.


\paragraph{Step 1: Asymptotic Expansions.} The asymptotic expansions of $\hat \thetavec$ and $\hat \phivec$ are derived in the following steps:
\begin{enumerate}
\item First-order asymptotic expansion of the first order conditions of \eqref{eq:est} around the true values $\thetavec^0$ and $\phivec^0$.
\item Second-order expansion of  $\hat \alpha_i(\thetavec^0,\phivec^0)$ from \eqref{eq:est2}.
\item Second-order expansion of the gradients $\sum_{t=1}^T \ell_{1it}^{\thetavec}(\thetavec^0,\hat \alpha_i(\thetavec^0,\phivec^0))$ and $\ell_{2i}^{\phivec}(\phivec^0,\hat \alpha_i(\thetavec^0,\phivec^0))$ with respect to $\hat \alpha_i(\thetavec^0,\phivec^0)$ around $\alpha_i^0$.
\item Asymptotic expansion of $\hat \thetavec$ and $\hat \phivec$.
\end{enumerate}

Assumption \ref{ass:gral} guarantees that in these expansions all the terms are bounded in probability and the remainders are negligible (e.g., Hahn and Newey, 2004, and Fern\'andez-Val, 2005).

\bigskip 

\noindent \underline{1. First-Order Expansion of the First-Order Conditions for $(\hat \thetavec, \hat \phivec)$:} the first-order conditions of  \eqref{eq:est} are
\begin{multline*}
0 = \frac{1}{NT} \sum_{i=1}^N \left\{ \sum_{t=1}^T  \ell_{1it}^{\thetavec}(\hat \thetavec,\hat \alpha_i(\hat \thetavec, \hat \phivec)) + \left[ \sum_{t=1}^T  \ell_{1it}^{\alpha}(\hat \thetavec,\hat \alpha_i(\hat \thetavec, \hat \phivec)) +  \ell_{2i}^{\alpha}(\hat \phivec,\hat \alpha_i(\hat \thetavec, \hat \phivec))   \right] \hat \alpha_i^{\thetavec}(\hat \thetavec, \hat \phivec) \right\} \\ =  \frac{1}{NT} \sum_{i=1}^N \sum_{t=1}^T  \ell_{1it}^{\thetavec}(\hat \thetavec,\hat \alpha_i(\hat \thetavec, \hat \phivec))
\end{multline*}
and
\begin{multline*}
0 = \frac{1}{NT} \sum_{i=1}^N \left\{  \ell_{2i}^{\phivec}(\hat \phivec,\hat \alpha_i(\hat \thetavec, \hat \phivec)) + \left[ \sum_{t=1}^T \ell_{1it}^{\alpha}(\hat \thetavec,\hat \alpha_i(\hat \thetavec, \hat \phivec)) +   \ell_{2i}^{\alpha}(\hat \phivec,\hat \alpha_i(\hat \thetavec, \hat \phivec)) \right] \hat \alpha_i^{\phivec}(\hat \thetavec, \hat \phivec)  \right\}  \\ 
= \frac{1}{NT} \sum_{i=1}^N \ell_{2i}^{\phivec}(\hat \phivec,\hat \alpha_i(\hat \thetavec, \hat \phivec)),
\end{multline*}
where the second equality in both cases  follows from 
\begin{equation}\label{eq:est2foc}
 \sum_{t=1}^T \ell_{1it}^{\alpha}( \thetavec,\hat \alpha_i( \thetavec,  \phivec)) + \ell_{2i}^{\alpha}( \phivec,\hat \alpha_i( \thetavec,  \phivec)) = 0, 
\end{equation}
by the first-order condition of \eqref{eq:est2}. 

A first-order expansion in $(\hat \thetavec, \hat \phivec)$ around $(\thetavec^0, \phivec^0)$ yields
\begin{multline}\label{eq:theta}
0 =   \frac{1}{NT} \sum_{i=1}^N \left[\begin{array}{c} \sum_{t=1}^T \ell_{1it}^{\thetavec}(\hat \thetavec,\hat \alpha_i(\hat \thetavec, \hat \phivec)) \\  \ell_{2i}^{\phivec}(\hat \phivec,\hat \alpha_i(\hat \thetavec, \hat \phivec)) \end{array}\right] = \frac{1}{NT} \sum_{i=1}^N \left[\begin{array}{c} \sum_{t=1}^T \ell_{1it}^{\thetavec}(\thetavec^0,\hat \alpha_i(\thetavec^0, \phivec^0)) \\  \ell_{2i}^{\phivec}(\phivec^0,\hat \alpha_i(\thetavec^0, \phivec^0)) \end{array}\right]  \\ + \frac{1}{NT} \sum_{i=1}^N  \left[\begin{array}{cc} \sum_{t=1}^T  \frac{d \ell_{1it}^{\thetavec}(\bar \thetavec, \hat \alpha_i(\bar \thetavec, \bar \phivec))}{d \thetavec^{\T}}  &  \sum_{t=1}^T  \frac{d \ell_{1it}^{\thetavec}(\bar \thetavec, \hat \alpha_i(\bar \thetavec, \bar \phivec))}{d \phivec^{\T}}  \\   \frac{d \ell_{2i}^{\phivec}(\bar \phivec, \hat \alpha_i(\bar \thetavec, \bar \phivec))}{d \thetavec^{\T}} &   \frac{d \ell_{2i}^{\phivec}(\bar \phivec, \hat \alpha_i(\bar \thetavec, \bar \phivec))}{d \phivec^{\T}} \end{array}\right] \left[\begin{array}{c}  \hat \thetavec - \thetavec^0 \\\hat \phivec - \phivec^0 \end{array}\right],
\end{multline}
%
where
$$
 \frac{d \ell_{1it}^{\thetavec}(\thetavec, \hat \alpha_i(\thetavec, \phivec))}{d \thetavec^{\T}} = \ell_{1it}^{\thetavec \thetavec}(\thetavec, \hat \alpha_i(\thetavec, \phivec)) + \ell_{1it}^{\thetavec \alpha}(\thetavec, \hat \alpha_i(\thetavec, \phivec) ) \hat \alpha_i^{\thetavec}(\thetavec, \phivec)^{\T},
$$
$$
 \frac{d \ell_{1it}^{\thetavec}(\thetavec, \hat \alpha_i(\thetavec, \phivec))}{d \phivec^{\T}} =  \ell_{1it}^{\thetavec \alpha}(\thetavec, \hat \alpha_i(\thetavec, \phivec) ) \hat \alpha_i^{\phivec}(\thetavec, \phivec)^{\T},
$$
$$
 \frac{d \ell_{2i}^{\phivec}(\phivec, \hat \alpha_i(\thetavec, \phivec))}{d \thetavec^{\T}} =  \ell_{2i}^{\phivec \alpha}(\phivec, \hat \alpha_i^{\thetavec}(\thetavec, \phivec) ) \hat \alpha_i^{\thetavec}(\thetavec, \thetavec)^{\T},
$$
$$
 \frac{d \ell_{2i}^{\phivec}(\phivec, \hat \alpha_i(\thetavec, \phivec))}{d \phivec^{\T}} = \ell_{2i}^{\phivec \phivec}(\phivec, \hat \alpha_i(\thetavec, \phivec)) + \ell_{2i}^{\phivec \alpha}(\phivec, \hat \alpha_i^{\thetavec}(\thetavec, \phivec) ) \hat \alpha_i^{\phivec}(\thetavec, \phivec)^{\T},
$$
and $(\bar \thetavec, \bar \phivec)$ lies between  $(\hat \thetavec, \hat \phivec)$ and $(\thetavec^0, \phivec^0)$. 

The expressions of $\hat \alpha_i^{\thetavec}(\thetavec, \phivec)$ and $ \hat \alpha_i^{\phivec}(\thetavec, \phivec)$ can be obtained from \eqref{eq:est2foc}.
Thus, differentiation with respect to $\thetavec$ and $\phivec$ gives
$$
\hat \alpha_i^{\thetavec}(\thetavec, \phivec) = - \frac{\sum_{t=1}^T \ell_{1it}^{\alpha \thetavec}(\thetavec, \hat \alpha_i(\thetavec, \phivec))^{\T}}{\sum_{t=1}^T \ell_{1it}^{\alpha \alpha}(\thetavec, \hat \alpha_i(\thetavec, \phivec)) +   \ell_{2i}^{\alpha \alpha}(\phivec, \hat \alpha_i(\thetavec, \phivec))},
$$
and
$$
\hat \alpha_i^{\phivec}(\thetavec, \phivec) = - \frac{ \ell_{2i}^{\alpha \phivec}(\phivec, \hat \alpha_i(\thetavec, \phivec))^{\T}}{\sum_{t=1}^T \ell_{1it}^{\alpha \alpha}(\thetavec, \hat \alpha_i(\thetavec, \phivec)) +   \ell_{2i}^{\alpha \alpha}(\phivec, \hat \alpha_i(\thetavec, \phivec))},
$$
respectively. Note that $\hat \alpha_i^{\thetavec}(\thetavec, \phivec) = O_P(1)$ and $\hat \alpha_i^{\phivec}(\thetavec, \phivec) = O_P(T^{-1})$, because all the terms are bounded in probability and the denominator is bounded away from zero with probability one. 

\bigskip

\noindent \underline{2. Second-Order Expansion of  $\hat \alpha_i(\thetavec^0,\phivec^0)$:} the first-order condition of  \eqref{eq:est2} at $(\thetavec,\phivec) = (\thetavec^0,\phivec^0)$ is 
\begin{equation*}
0 = \sum_{t=1}^T \ell_{1it}(\thetavec^0, \hat \alpha_i(\thetavec^0,\phivec^0)) +   \ell_{2i}(\phivec^0, \hat \alpha_i(\thetavec^0,\phivec^0)).
\end{equation*}
A second-order expansion in $\hat \alpha_i(\thetavec^0,\phivec^0)$ around $\alpha_i^0$ yields (e.g., Rilstone, Srivastava and Ullah, 1996)
\begin{equation}\label{eq:alpha}
\hat \alpha_i(\thetavec^0,\phivec^0) - \alpha_i^0 = \frac{\psi_i}{\sqrt{T}} + \frac{\beta_i}{T} + o_P(T^{-1}) 
\end{equation}
where the expressions of the influence function $\psi_i$ and second-order bias  $\beta_i$ are given in Assumption \ref{ass:gral}. To bound the remainder term uniformly over $i$, we use that
$$
\sup_i \|\hat \alpha_i(\thetavec^0, \phivec^0) - \alpha_i^0 \| \leq \sup_i \|\hat \alpha_i(\thetavec^0, \phivec^0) -  \hat \alpha_i(\hat \thetavec, \hat \phivec)\| + \sup_i \|\hat \alpha_i - \alpha_i^0 \| \to_P 0,
$$
by the triangle inequality, a.s. continuity of $(\thetavec, \phivec) \mapsto \hat \alpha_i( \thetavec,  \phivec)$ and Assumption \ref{ass:gral}(i).

\bigskip

\noindent \underline{3. Second-Order Expansion of the Gradients $\sum_{t=1}^T \ell_{1it}^{\thetavec}(\thetavec^0,\hat \alpha_i(\thetavec^0,\phivec^0))$ and $\ell_{2i}^{\phivec}(\phivec^0,\hat \alpha_i(\thetavec^0,\phivec^0))$:} a second-order expansion of $\sum_{t=1}^T \ell_{1it}^{\thetavec}(\thetavec^0,\hat \alpha_i(\thetavec^0,\phivec^0))$ in $\hat \alpha_i(\thetavec^0,\phivec^0)$ around $\alpha_i^0$ yields
$$
\frac{1}{T} \sum_{t=1}^T \ell_{1it}^{\thetavec}(\thetavec^0,\hat \alpha_i(\thetavec^0,\phivec^0)) = \frac{\varphi_{\thetavec i} }{\sqrt{T}} + \frac{b_{\thetavec i}}{T} + o_P(T^{-1}),
$$
where the expressions of the influence function $\varphi_{\thetavec i}$, and  second-order bias $b_{\thetavec i}$ are given in Assumption \ref{ass:gral}, and we use \eqref{eq:alpha}. A similar analysis for  $ \ell_{2i}^{\phivec}(\phivec^0,\hat \alpha_i(\thetavec^0,\phivec^0))$ gives
$$
\ell_{2i}^{\phivec}(\phivec^0,\hat \alpha_i(\thetavec^0,\phivec^0)) = \varphi_{\phivec i}  + \frac{b_{\phivec i}}{T} + o_P(T^{-1}),
$$
where the expression of $\varphi_{\phivec i}$ and $b_{\phivec i}$ are given in Assumption \ref{ass:gral}, and we use again \eqref{eq:alpha}.

\bigskip

\paragraph{Step 2: Asymptotic Distributions of $\hat \thetavec$ and $\hat \phivec$.} 
From the expansions of the gradients in the previous section
$$
 \sum_{i=1}^N \sum_{t=1}^T \ell_{1it}^{\thetavec}(\thetavec^0,\hat \alpha_i(\thetavec^0,\phivec^0)) = O_P(\sqrt{NT} \vee N),
$$
and
$$
\sum_{i=1}^N \ell_{2i}^{\phivec}(\phivec^0,\hat \alpha_i(\thetavec^0,\phivec^0)) = O_P(\sqrt{N}),
$$
if $N = O(T^{2})$.  Here we use that $\sum_{i=1}^N \varphi_{\thetavec i} = O_P(\sqrt{N})$, $\sum_{i=1}^N b_{\thetavec i} = O_P(N)$, $\sum_{i=1}^N \varphi_{\phivec i} = O_P(\sqrt{N})$, and $\sum_{i=1}^N b_{\phivec i} = O_P(N)$ (e.g., Fern\'andez-Val, 2005).

We need to consider two different cases because the asymptotic distribution of  $\hat \phivec$ is non-degenerate under a wider range of sequences for $N$ and $T$ than the distribution of $\hat \thetavec$.  These cases are
\begin{enumerate}
\item $N = O(T)$: combining the previous results with \eqref{eq:theta} yields
$
\sqrt{NT}(\hat \thetavec - \thetavec^0) = O_P(1)
$
and 
$
\sqrt{N}(\hat \phivec - \phivec^0) = O_P(1),
$
such that
\begin{multline}\label{eq:theta2} 
\frac{1}{N}\sum_{i=1}^N  \left[\begin{array}{cc} \frac{1}{T} \sum_{t=1}^T  \frac{d \ell_{1it}^{\thetavec}(\bar \thetavec, \hat \alpha_i(\bar \thetavec, \bar \phivec))}{d \thetavec^{\T}}  &  \frac{1}{\sqrt{T}} \sum_{t=1}^T  \frac{d \ell_{1it}^{\thetavec}(\bar \thetavec, \hat \alpha_i(\bar \thetavec, \bar \phivec))}{d \phivec^{\T}}  \\   \frac{1}{\sqrt{T}} \frac{d \ell_{2i}^{\phivec}(\bar \phivec, \hat \alpha_i(\bar \thetavec, \bar \phivec))}{d \thetavec^{\T}} &   \frac{d \ell_{2i}^{\phivec}(\bar \phivec, \hat \alpha_i(\bar \thetavec, \bar \phivec))}{d \phivec^{\T}} \end{array}\right] \left[\begin{array}{c} \sqrt{NT}(\hat \thetavec - \thetavec^0) \\\sqrt{N}(\hat \phivec - \phivec^0)\end{array}\right] \\ = \frac{-1}{\sqrt{N}} \sum_{i=1}^N \left[\begin{array}{c} \frac{1}{\sqrt{T}}\sum_{t=1}^T \ell_{1it}^{\thetavec}(\thetavec^0,\hat \alpha_i(\thetavec^0, \phivec^0)) \\  \ell_{2i}^{\phivec}(\phivec^0,\hat \alpha_i(\thetavec^0, \phivec^0)) \end{array}\right].
\end{multline}
Note that  the off-diagonal elements of the matrix in the LHS are of order $O_P(T^{-1})$, whereas the diagonal elements are of order $O_P(1)$. Hence
\begin{multline*}
\frac{1}{N}\sum_{i=1}^N  \left[\begin{array}{cc} \frac{1}{T} \sum_{t=1}^T  \frac{d \ell_{1it}^{\thetavec}(\bar \thetavec, \hat \alpha_i(\bar \thetavec, \bar \phivec))}{d \thetavec^{\T}}  &  \frac{1}{\sqrt{T}} \sum_{t=1}^T  \frac{d \ell_{1it}^{\thetavec}(\bar \thetavec, \hat \alpha_i(\bar \thetavec, \bar \phivec))}{d \phivec^{\T}}  \\   \frac{1}{\sqrt{T}} \frac{d \ell_{2i}^{\phivec}(\bar \phivec, \hat \alpha_i(\bar \thetavec, \bar \phivec))}{d \thetavec^{\T}} &   \frac{d \ell_{2i}^{\phivec}(\bar \phivec, \hat \alpha_i(\bar \thetavec, \bar \phivec))}{d \phivec^{\T}} \end{array}\right] \\ = \frac{1}{N}\sum_{i=1}^N  \left[\begin{array}{cc} \frac{1}{T} \sum_{t=1}^T  \frac{d \ell_{1it}^{\thetavec}(\bar \thetavec, \hat \alpha_i(\bar \thetavec, \bar \phivec))}{d \thetavec^{\T}}  &  0 \\  0 &   \frac{d \ell_{2i}^{\phivec}(\bar \phivec, \hat \alpha_i(\bar \thetavec, \bar \phivec))}{d \phivec^{\T}} \end{array}\right] + O_P\left( \frac{1}{\sqrt{T}}\right). 
\end{multline*}

\item $T = o(N): $ combining the previous results with \eqref{eq:theta} yields
$
T(\hat \thetavec - \thetavec^0) = O_P(1)
$
and 
$
\sqrt{N}(\hat \phivec - \phivec^0) = O_P(1),
$
such that
\begin{multline}\label{eq:theta2b} 
\frac{1}{N}\sum_{i=1}^N  \left[\begin{array}{cc} \frac{1}{T} \sum_{t=1}^T  \frac{d \ell_{1it}^{\thetavec}(\bar \thetavec, \hat \alpha_i(\bar \thetavec, \bar \phivec))}{d \thetavec^{\T}}  &  \frac{1}{\sqrt{N}} \sum_{t=1}^T  \frac{d \ell_{1it}^{\thetavec}(\bar \thetavec, \hat \alpha_i(\bar \thetavec, \bar \phivec))}{d \phivec^{\T}}  \\   \frac{\sqrt{N}}{T} \frac{d \ell_{2i}^{\phivec}(\bar \phivec, \hat \alpha_i(\bar \thetavec, \bar \phivec))}{d \thetavec^{\T}} &   \frac{d \ell_{2i}^{\phivec}(\bar \phivec, \hat \alpha_i(\bar \thetavec, \bar \phivec))}{d \phivec^{\T}} \end{array}\right] \left[\begin{array}{c} T(\hat \thetavec - \thetavec^0) \\\sqrt{N}(\hat \phivec - \phivec^0)\end{array}\right] \\ = \frac{-1}{\sqrt{N}} \sum_{i=1}^N \left[\begin{array}{c} \frac{1}{\sqrt{N}}\sum_{t=1}^T \ell_{1it}^{\thetavec}(\thetavec^0,\hat \alpha_i(\thetavec^0, \phivec^0)) \\  \ell_{2i}^{\phivec}(\phivec^0,\hat \alpha_i(\thetavec^0, \phivec^0)) \end{array}\right].
\end{multline}
Similar to the other case,  some of the off-diagonal elements of the matrix in the LHS are of smaller order than the diagonal elements. Hence, if $N = O(T^2)$,
\begin{multline*}
\frac{1}{N}\sum_{i=1}^N  \left[\begin{array}{cc} \frac{1}{T} \sum_{t=1}^T  \frac{d \ell_{1it}^{\thetavec}(\bar \thetavec, \hat \alpha_i(\bar \thetavec, \bar \phivec))}{d \thetavec^{\T}}  &  \frac{1}{\sqrt{N}} \sum_{t=1}^T  \frac{d \ell_{1it}^{\thetavec}(\bar \thetavec, \hat \alpha_i(\bar \thetavec, \bar \phivec))}{d \phivec^{\T}}  \\   \frac{\sqrt{N}}{T} \frac{d \ell_{2i}^{\phivec}(\bar \phivec, \hat \alpha_i(\bar \thetavec, \bar \phivec))}{d \thetavec^{\T}} &   \frac{d \ell_{2i}^{\phivec}(\bar \phivec, \hat \alpha_i(\bar \thetavec, \bar \phivec))}{d \phivec^{\T}} \end{array}\right] \\ = \frac{1}{N}\sum_{i=1}^N  \left[\begin{array}{cc} \frac{1}{T} \sum_{t=1}^T  \frac{d \ell_{1it}^{\thetavec}(\bar \thetavec, \hat \alpha_i(\bar \thetavec, \bar \phivec))}{d \thetavec^{\T}}  &  0 \\  \frac{\sqrt{N}}{T} \frac{d \ell_{2i}^{\phivec}(\bar \phivec, \hat \alpha_i(\bar \thetavec, \bar \phivec))}{d \thetavec^{\T}}  &   \frac{d \ell_{2i}^{\phivec}(\bar \phivec, \hat \alpha_i(\bar \thetavec, \bar \phivec))}{d \phivec^{\T}} \end{array}\right] + o_P\left( 1\right). 
\end{multline*}

\end{enumerate}

Note that 
\begin{eqnarray*}
& &\plim_{N,T \to \infty} \frac{1}{NT} \sum_{i=1}^N \sum_{t=1}^T  \frac{d \ell_{1it}^{\thetavec}(\bar \thetavec, \hat \alpha_i(\bar \thetavec, \bar \phivec))}{d \thetavec^{\T}}  = \plim_{N,T \to \infty} \frac{1}{NT} \sum_{i=1}^N \sum_{t=1}^T \left[ \ell_{1it}^{\thetavec \thetavec} - \ell_{1it}^{\thetavec \alpha} \frac{\sum_{t=1}^T \ell_{1it}^{\alpha \thetavec } }{\sum_{t=1}^T  \ell_{1it}^{\alpha \alpha }} \right], \\
& &\plim_{N,T \to \infty} \frac{1}{N} \sum_{i=1}^N   \frac{d \ell_{2i}^{\phivec}(\bar \thetavec, \hat \alpha_i(\bar \thetavec, \bar \phivec))}{d \thetavec^{\T}}  = - \plim_{N,T \to \infty} \frac{1}{N} \sum_{i=1}^N \ell_{2i}^{\phivec \alpha} \frac{\sum_{t=1}^T \ell_{1it}^{\alpha \thetavec } }{\sum_{t=1}^T  \ell_{1it}^{\alpha \alpha } }, \\
& &\plim_{N \to \infty} \frac{1}{N} \frac{d \ell_{2i}^{\phivec}(\bar \phivec, \hat \alpha_i(\bar \thetavec, \bar \phivec))}{d \phivec^{\T}} = \plim_{N \to \infty} \frac{1}{N} \sum_{i=1}^N  \ell_{2i}^{\phivec \phivec},
\end{eqnarray*}
because $
\sqrt{NT}(\bar \thetavec - \thetavec^0) = O_P(1),
$
$
\sqrt{N}(\bar \phivec - \phivec^0) = O_P(1),
$
and
$$
 \hat \alpha_i(\bar \thetavec, \bar \phivec) = \hat \alpha_i(\thetavec^0, \phivec^0) +  \hat \alpha_i^{\thetavec}(\check \thetavec, \check \phivec)^{\T}(\bar \thetavec - \thetavec^0) + \hat \alpha_i^{\phivec}(\check \thetavec, \check \phivec)^{\T}(\bar \phivec - \phivec^0) = \alpha_i^0 + O_P(T^{-1/2}),
$$
where $(\check \thetavec, \check \phivec)$ lies between  $(\bar \thetavec, \bar \phivec)$ and $(\thetavec^0, \phivec^0)$.


Starting from \eqref{eq:theta2} and combining the results from the previous steps yields
$$
0 = \frac{1}{\sqrt{N}} \sum_{i=1}^N \varphi_{\thetavec i} + \bar{\H}_{\thetavec} \sqrt{NT}\left( \hat \thetavec - \thetavec^0 + \frac{\bar{\H}_{\thetavec}^{-1} \bar{b}_{\thetavec}}{T} \right) + o_P\left( \sqrt{\frac{N}{T}}\right) + o_P\left(\sqrt{NT}( \hat \thetavec - \thetavec^0 )\right)
$$
Hence, under asymptotic sequences where $N = O(T)$,
$$
\sqrt{NT}\left( \hat \thetavec - \thetavec^0 + \frac{\bar{\H}_{\thetavec}^{-1} \bar{b}_{\thetavec}}{T} \right) \to_d \bar{\H}_{\thetavec}^{-1} \mathcal{N}(0, \bar{\O}_{\thetavec}). 
$$ 
%
%
Similarly, starting from \eqref{eq:theta2b}, combining the results from the previous steps yields and using block matrix inversion
\begin{multline*}
0 = \frac{1}{\sqrt{N}} \sum_{i=1}^N \varphi_{\phivec i} + \bar{\H}_{\phivec} \sqrt{N}\left( \hat \phivec - \phivec^0 + \frac{\bar{\H}_{\phivec}^{-1} (\bar{b}_{\phivec} +  \bar{\H}_{\phivec\thetavec} \bar{\H}_{\thetavec}^{-1}\bar{b}_{\thetavec})}{T} \right) \\ + o_P\left( \frac{\sqrt{N}}{T} \right)  + o_P\left(\sqrt{N}( \hat \phivec - \phivec^0 )\right).
\end{multline*}
Hence, under asymptotic sequences where $N = O(T^2)$,
$$
\sqrt{N}\left( \hat \phivec - \phivec^0 + \frac{\bar{\H}_{\phivec}^{-1} (\bar{b}_{\phivec} +  \bar{\H}_{\phivec\thetavec} \bar{\H}_{\thetavec}^{-1}\bar{b}_{\thetavec})}{T} \right) \to_d \bar{\H}_{\phivec}^{-1} \mathcal{N}(0, \bar{\O}_{\phivec}). 
$$  
\end{proof}

\subsection*{Application to ILM Estimator}

In the case of the integrated loss minimization (ILM) estimator for longitudinal data
\begin{equation}\label{eq:l1}
 \ell_{1it}(\thetavec,\alpha_i) = (y_{it} - \alpha_i)(u_{it}(\thetavec,\alpha_i) -0.5 ) + \xx_{it}^{\T} \thetavec [\bar \Bvec - \Bvec(u_{it}(\thetavec,\alpha_i))],
\end{equation}
where $\Bvec(u) = \int_0^u \bvec(x) dx$, $\bar \Bvec = \int_0^1 \Bvec(u) du$, and $u_{it}(\thetavec,\alpha_i)$ is the solution in $u$ to
$$
y_{it} - \alpha_i = \xx_{it}^{\T} \thetavec \bvec(u);
$$
and 
\begin{equation}\label{eq:l2}
 \ell_{2i}(\phivec,\alpha_i) =  \alpha_i(v_{i}(\phivec,\alpha_i) -0.5 ) + \zz_{i}^{\T} \phivec [\bar \Cvec - \Cvec(v_i(\phivec,\alpha_i))],
\end{equation} 
where $\Cvec(v) = \int_0^v \cvec(x) dx$, $\bar \Cvec = \int_0^1 \Cvec(v) dv$, and $v_{i}(\phivec,\alpha_i)$ is the solution in $v$ to
$$
\alpha_i = \zz_{i}^{\T} \phivec \cvec(v).
$$
To apply the previous analysis to the ILM estimator, we need to adapt the notation to the case where $\thetavec$ and $\phivec$ can be matrix-valued parameters by vectorizing $\thetavec$ and $\phivec$ in all the expressions. For example, $ \ell_{1it}^{\thetavec}(\thetavec, \alpha_i)$ becomes $\partial  \ell_{1it}(\thetavec, \alpha_i) / \partial \ve(\thetavec)$ and $(\hat \phivec - \phivec^0)$ becomes $\ve(\hat \phivec - \phivec^0)$.   

\begin{proof} [Proof of Theorem 1] 
%
We divide the proof in three parts: identification,  consistency, and asymptotic distribution.

\paragraph{Part 1: Identification.} The identification analysis has several steps. First, we show identification of the quantile regression level-1 coefficient functions of the time-varying covariates and an aggregated individual effect that contains the quantile regression level-1 coefficient functions of the time invariant covariates and individual effects, using within-group variation. Second, we separate the quantile regression level-1 coefficient functions of the time invariant covariates from the individual effects using between-group variation and the location normalization in the distribution of the individual effects. Third, we show identification of the quantile regression level-2 coefficient functions using between-group variation of the individual effects. Fourth, we show that the parametric coefficient functions are identified from the quantile regression coefficient functions as the solutions to linear systems of equations. These solutions exist and are unique by Assumption 1(iv).

In this part, it is convenient to partition $\xx_{it} = (\xx_{1i},\xx_{2it})$, where $\xx_{1i}$ contains the time-invariant components including the constant and $\xx_{2it}$ contains the time-varying covariates,  $\thetavec = (\thetavec_1,\thetavec_2)$ and $\bvec(u) = [\bvec_1(u), \bvec_2(u)]$,  such that $\xx_{it}^{\T} \thetavec \bvec(u) = \xx_{1i}^{\T} \thetavec_1 \bvec_1(u) + \xx_{2it}^{\T} \thetavec_2 \bvec_2(u)$; and  use the parametrization $\alpha_i(u) = \alpha_i + \xx_{1i}^{\T}\thetavec_1 \bvec_1(u)$, $\thetavec_2(u) := \thetavec_2 \bvec(u)$ and $\phivec(v) := \phivec \cvec(v)$. Then, we can express the elements of the objective function as
$$
\ell_{1it}(\thetavec,\alpha_i) = \int_0^1 \rho_u(y_{it} - \alpha_i(u) - \xx_{2it}^{\T} \thetavec_2(u)) du,
$$
and
$$
\ell_{2i}(\phivec,\alpha_i) = \int_0^1  \rho_v(\alpha_i - \zz_i^{\T} \phivec(v)) dv,
$$
where $\rho_{\tau}(t) = t(\tau - I(t \leq 0))$ for any $\tau \in [0,1]$. 

For any $u$, identification of $\thetavec_2^0(u) = \thetavec_2^0 \bvec_2(u)$ and $\alpha_i^0(u) = \alpha_i^0 + \xx_{1i}^{\T}\thetavec_1^0 \bvec_1(u)$ follows from Assumption 1(i)--(iii) by standard arguments for quantile regression since 
$$
f_{Y_{it}}(\alpha_i^0(u) + \xx_{2it}^{\T} \thetavec_2^0(u) \mid \alpha_i, \xx_{it}) = \frac{1}{\xx_{it}^{\T}\thetavec^0 \bvec'(u)} < \infty,
$$
where $f_{Y_{it}}(\cdot \mid \alpha_i, \xx_{it}) $ is the distribution of $Y_{it}$ conditional on $\alpha_i$ and $\xx_{it}$.
Then, to identify $\alpha_i^0$ and $\thetavec_1^0(u) = \thetavec_1^0 \bvec_1(u)$ from $\alpha_i^0(u)$ we use the within expectation
$$
\Ep[\alpha_i^0(u) \mid \alpha_i, \xx_{1i}] = \alpha_i^0 + \xx_{1i}^{\T} \Ep[\thetavec_1(U_{it})].
$$  
The between variation of this expression identifies $\plim_{N \to \infty} N^{-1} \sum_{i=1}^N \alpha_i^0 + \Ep(\theta_{1}(U_{it}))$ and  $\Ep[\thetavec_{1,-1}(U_{it})]$, where $\theta_1(u)$  and $\thetavec_{1,-1}(u)$ are the intercept and vector of slopes of $\thetavec_1(u)$, respectively.  To separate $\plim_{N \to \infty} N^{-1} \sum_{i=1}^N \alpha_i^0$ from  $\Ep(\theta_{1}(U_{it}))$ we use the normalization in Assumption 1(i). Thus, 
$$
\plim_{N\to \infty} \frac{1}{N} \sum_{i=1}^N \alpha_i^0(u) = \plim_{N\to \infty}\frac{1}{N}  \sum_{i=1}^N \alpha_i^0 + \plim_{N\to \infty}\frac{1}{N}  \sum_{i=1}^N \xx_{1i}^{\T} \Ep[\thetavec_1(U_{it})] 
$$
pins down $\Ep[\theta_1(U_{it})]$ since $\plim_{N\to \infty}N^{-1} \sum_{i=1}^N \alpha_i^0 = 0$. Then, $\alpha_i^0$ is identified by
$$
\alpha_i^0 = \Ep[\alpha_i^0(u) \mid \alpha_i, \xx_{1i}] - \xx_{1i}^{\T} \Ep[\thetavec_1(U_{it})],
$$
and $\thetavec_{1}(u)$ is identified by
$$
\thetavec_{1}(u) = \left[ \plim_{N\to \infty}\frac{1}{N}  \sum_{i=1}^N  \xx_{1i} \xx_{1i}^{\T} \right]^{-1} \plim_{N\to \infty}\frac{1}{N}  \sum_{i=1}^N \xx_{1i} (\alpha_i^0(u) - \alpha_i^0).
$$
Finally, identification of $\phivec^0(v) = \phivec^0 \cvec(v)$ for any $v$ follows from identification of $\alpha_i^0$ by standard arguments for quantile regression since
$$
f_{\alpha_i}(\zz_{i}^{\T} \phivec^0(v) \mid \zz_{i}) = \frac{1}{\zz_{i}^{\T}\phivec^0 \cvec'(v_{i})} < \infty,
$$
where $f_{\alpha_i}(\cdot \mid \zz_{i})$ is the distribution of $\alpha_i$ conditional on $\zz_{i}$.

To show identification of $\thetavec^0$ and $\phivec^0$ from identification of $\thetavec^0(u) = [\thetavec_1^0(u), \thetavec_2^0(u)]$ and $\phivec^0(v)$, we use Assumption 1(iv). 
Let $\thetavec^0(u_1, \ldots, u_{d_{b}}) = [\thetavec^0(u_1), \ldots, \thetavec^0(u_{d_{b}})]$ and $\phivec^0(v_1, \ldots, v_{d_{c}}) = [\phivec^0(v_1), \ldots, \phivec^0(v_{d_{c}})]$. Then, we have two systems of linear equations
$$
\thetavec^0(u_1, \ldots, u_{d_{b}})  = \thetavec^0 \bvec(u_1, \ldots, u_{d_{b}}), \ \ \phivec^0(v_1, \ldots, v_{d_{c}}) = \phivec^0 \cvec(v_1, \ldots, v_{d_{c}}), 
$$
which have as unique solutions
$$
\thetavec^0 = \thetavec^0(u_1, \ldots, u_{d_{b}}) \bvec(u_1, \ldots, u_{d_{b}})^{-1}, \ \ \phivec^0 = \phivec^0(v_1, \ldots, v_{d_{c}}) \cvec(v_1, \ldots, v_{d_{c}})^{-1}.
$$


\paragraph{Part 2: Consistency.} The consistency of all the estimators can be established sequentially. Let
$$
 \hat L_{i}(\thetavec, \phivec, \alpha_i) = \frac{1}{T} \left[\sum_{t=1}^T \ell_{1it}(\thetavec, \alpha_i) +   \ell_{2i}(\phivec, \alpha_i)\right],
$$
and
$$
L_{1i}(\thetavec, \alpha_i) = \plim_{T \to \infty} \frac{1}{T} \sum_{t=1}^T \ell_{1it}(\thetavec, \alpha_i). 
$$
Consistency of $\hat \thetavec$ and uniform consistency of $\hat \alpha_i$, i.e., $\hat \thetavec \to_P \thetavec^0$ and $\sup_i |\hat \alpha_i -  \alpha_i^0| \to_P 0$,  follows from standard arguments for quantile regression using the parametrization of step 1, together with the fact that the penalty term of the objective function is asymptotically negligible, i.e.,
$$
\sup_i \left[ \sup_{\{\thetavec, \phivec, \alpha_i\}} \left| \hat L_{i}(\thetavec, \phivec, \alpha_i) - L_{1i}(\thetavec, \alpha_i) \right| \right] \to_P 0,
$$
where we use that
$$
\sup_i \left[ \sup_{\{\phivec, \alpha_i\}} \left| \frac{1}{T} \ell_{2i}(\phivec,\alpha_i) \right| \right] \to_P 0,
$$
and the triangle inequality.

Then, $\hat \phivec \to_P \phivec^0$ also follows from standard arguments for quantile regression using that
$$
\sup_{\phivec} \left|  \frac{1}{N} \sum_{i=1}^N \ell_{2i}(\phivec, \hat \alpha_i) - L_2(\phivec)\right| \to_P 0, \ \  L_2(\phivec) = \plim_{N \to \infty} \frac{1}{N}  \sum_{i=1}^N \ell_{2i}(\phivec, \alpha_i^0),
$$
where we use that, by the triangle inequality,
\begin{multline*}
\left|  \frac{1}{N} \sum_{i=1}^N \ell_{2i}(\phivec, \hat \alpha_i) - L_2(\phivec)\right| \leq \left|  \frac{1}{N} \sum_{i=1}^N \ell_{2i}(\phivec, \hat \alpha_i) - \frac{1}{N} \sum_{i=1}^N \ell_{2i}(\phivec,  \alpha_i^0) \right| \\ + \left|  \frac{1}{N} \sum_{i=1}^N \ell_{2i}(\phivec, \alpha_i^0) - L_2(\phivec)\right|, 
\end{multline*}
continuity of $\alpha_i \mapsto \ell_{2i}(\phivec,  \alpha_i)$,  and uniform consistency of $\hat \alpha_i$.

\paragraph{Part 3: Asymptotic Distribution. } The derivation of the asymptotic distribution follows from Theorem \ref{thm:gral} after verifying Assumption \ref{ass:gral}(ii)-(viii) and evaluating the expressions for the ILM estimator.  

We start by obtaining the derivatives of $(\thetavec,\alpha_i) \mapsto \ell_{1it}(\thetavec,\alpha_i)$ and $\phivec \mapsto \ell_{2i}(\phivec,\alpha_i)$. We report only the derivatives that show up in the  terms of the asymptotic expansions. A similar analysis applies to the additional terms that appear in the remainder terms. Direct calculations from \eqref{eq:l1} yield
\begin{eqnarray*}
 \ell_{1it}^{\alpha}(\thetavec,\alpha_i) &=& - (u_{it}(\thetavec,\alpha_i) -0.5 ), \ \  \ell_{1it}^{\alpha \alpha}(\thetavec,\alpha_i) = - u_{it}^{\alpha}(\thetavec,\alpha_i) = [ \xx_{it}^{\T} \thetavec  \bvec'(u_{it}(\thetavec,\alpha_i)]^{-1}, \\ \ell_{1it}^{\alpha \alpha \alpha}(\thetavec,\alpha_i) &=&  - \xx_{it}^{\T} \thetavec \bvec''(u_{it}(\thetavec,\alpha_i)) \ell_{1it}^{\alpha \alpha}(\thetavec,\alpha_i)^{3}, \\
  \ell_{1it}^{\thetavec}(\thetavec,\alpha_i) &=&  [\bar \Bvec - \Bvec(u_{it}(\thetavec,\alpha_i))] \otimes \xx_{it}, \\   
  \ell_{1it}^{\thetavec \thetavec}(\thetavec,\alpha_i) &=&  [ \ell_{1it}^{\alpha \alpha}(\thetavec,\alpha_i) \bvec(u_{it}(\thetavec,\alpha_i)) \bvec(u_{it}(\thetavec,\alpha_i))^{\T}] \otimes  [\xx_{it} \xx_{it}^{\T}], \\ 
   \ell_{1it}^{\thetavec \alpha}(\thetavec,\alpha_i) &=&   [\ell_{1it}^{\alpha \alpha}(\thetavec,\alpha_i) \bvec(u_{it}(\thetavec,\alpha_i))] \otimes  \xx_{it}, \\  \ell_{1it}^{\thetavec \alpha \alpha}(\thetavec,\alpha_i) &=& [  \ell_{1it}^{\alpha \alpha \alpha}(\thetavec,\alpha_i) \bvec(u_{it}(\thetavec,\alpha_i)) - \ell_{1it}^{\alpha \alpha}(\thetavec,\alpha_i)^2  \bvec'(u_{it}(\thetavec,\alpha_i)) ] \otimes \xx_{it}.
\end{eqnarray*}
Analogously, direct calculations from \eqref{eq:l2} yield
\begin{eqnarray*}
 \ell_{2i}^{\alpha}(\phivec,\alpha_i) &=& v_i(\phivec,\alpha_i) -0.5 , \ \  \ell_{2i}^{\alpha \alpha}(\phivec,\alpha_i) =  v_i^{\alpha}(\phivec,\alpha_i) =  [ \zz_{i}^{\T} \phivec  \cvec'(v_i(\phivec,\alpha_i)]^{-1}, \\ 
 \ell_{2i}^{\alpha \alpha \alpha}(\phivec,\alpha_i) &=&  - \zz_{i}^{\T} \phivec \cvec''(v_i(\phivec,\alpha_i)) \ell_{2i}^{\alpha \alpha}(\phivec,\alpha_i)^3, \\
  \ell_{2i}^{\phivec}(\phivec,\alpha_i) &=&  [\bar \Cvec - \Cvec(v_i(\phivec,\alpha_i))] \otimes \zz_{i}, \\   \ell_{2i}^{\phivec \phivec}(\phivec,\alpha_i) &=&  [\ell_{2i}^{\alpha \alpha}(\phivec,\alpha_i) \cvec(v_i(\phivec,\alpha_i)) \cvec(v_i(\phivec,\alpha_i))^{\T}] \otimes [\zz_{i} \zz_{i}^{\T}], \\ 
   \ell_{2i}^{\phivec \alpha}(\phivec,\alpha_i) &=& -  [\cvec(v_i(\phivec,\alpha_i)) \ell_{2i}^{\alpha \alpha}(\phivec,\alpha_i) ] \otimes  \zz_{i}, \\  \ell_{2i}^{\phivec \alpha \alpha}(\phivec,\alpha_i) &=& -  [ \ell_{2i}^{\alpha \alpha \alpha}(\phivec,\alpha_i) \cvec(v_i(\phivec,\alpha_i))   + \ell_{2i}^{\alpha \alpha}(\phivec,\alpha_i)^2 \cvec'(v_i(\phivec,\alpha_i)) ] \otimes \zz_{i}.
\end{eqnarray*}

Assumption \ref{ass:gral}(ii) follows from Assumption 1(v) 
by inspection of the derivatives. Assumption \ref{ass:gral}(iii) holds trivially as the parameter spaces are $\mathcal{A} = \mathbb{R}$, $\Theta = \mathbb{R}^{d_{x}}$ and $\Phi = \mathbb{R}^{d_{z}}$, and Assumptions \ref{ass:gral}(iv)-(v) follow directly from Assumption 1(vi)-(vii). 
Then, the asymptotic distributions follow from Theorem \ref{thm:gral}, replacing the derivatives above evaluated at the true parameter values in the expressions of the bias and variance given in Assumption \ref{ass:gral}, after dropping out some terms that are either asymptotically negligible or zero because $u_{it} \sim U_{it}$ and $v_i \sim V_{i}$ are independent.
\end{proof}

\section*{Appendix B - Computation}

An efficient algorithm has been implemented in the \texttt{qrcm} R package, which also includes
all the necessary functions for model building, summary, predictions, testing, and plotting.
The steps of the algorithm can be summarized as follows:
\begin{itemize}
\item{step 0. Select starting values $\hat\csi^{(0)} = (\hat\thetavec^{(0)}, \hat\phivec^{(0)}, \aahat^{(0)}).$}
\item{step 1.
	Given a current estimate $\aahat$, minimize $L(\thetavec, \aahat)$ (equation 6) 
        with respect to $\thetavec$, and $L(\phivec, \aahat)$ (equation 7) 
        with respect to $\phivec$. Equivalently, find the approximated zeroes of 
	$\G_{\thetavec}(\thetavec, \aahat)$ (equation 13) 
       and $\G_{\phivec}(\phivec, \aahat)$ (equation 14). 
}
\item{step 2.
	Given a current estimate $(\hat\thetavec, \hat\phivec)$, find the approximated zeroes of 
	$G_{\alpha_i}(\alpha_i, \hat\thetavec, \hat\phivec)$ (equation 15), 
        for $i = 1, \ldots, N$.
}
\end{itemize}
Steps 1-2 are repeated until a convergence criterion has been reached. In the \texttt{qrcm} package, 
the algorithm stops when, in two consecutive iterations, either the absolute difference in the estimated parameters, 
or the absolute change in the loss function defined by $L(\hat\thetavec,\hat \phivec, \aahat) = L(\hat\thetavec, \aahat) + L(\hat\phivec, \aahat)$
(equation 12) 
is below a certain tolerance (default $0.00001$).

Note that $\G_{\thetavec}(\thetavec, \aa)$, $\G_{\phivec}(\phivec,\aa)$, and $G_{\alpha_i}(\alpha_i, \thetavec, \phivec)$
are smooth functions of their arguments, and can be solved by using a standard Newton-type algorithm.
In practice, a bisection algorithm is used to solve $G_{\alpha_i}(\alpha_i, \hat\thetavec, \hat\phivec)$
(see section {\bf B2} of this Appendix).

The table below summarizes the computation times required to estimate the models presented in simulation,
using a desktop computer Intel(R) Core(TM) i7-4770 CPU @ 3.40GHz, RAM 8.00 GB, 64-bit Operating System. 
\begin{table}[h]
\small
\centering
\begin{tabular}{ccccccc}
\hline
\hline
\noalign{\vspace{0.1cm}}
&&\multicolumn{2}{c}{$N = 150$} && \multicolumn{2}{c}{$N = 300$}\\
\cline{3-4}\cline{6-7}
\noalign{\vspace{0.1cm}}
&& $T = 5$ & $T = 10$ && $T = 5$ & $T = 10$\\
Simulation 1 && 0.7 (0.4--0.8) & 1.6 (0.8--2.2) && 1.3 (0.7--1.6) & 2.7 (1.3--3.8)\\
Simulation 2 && 1.7 (1.5--1.8) & 2.3 (2.1--2.5) && 3.2 (3.0--3.4) & 4.7 (4.4--5.0)\\
\noalign{\vspace{0.1cm}}
\hline
\hline
\noalign{\vspace{0.1cm}}
\end{tabular}
\vspace{0.1cm}\footnotesize\center
Summary statistics of computation times for simulations 1 and 2 described in Section \ref{sec:sim}.
We report the median time (in seconds) and, in brackets, the interquartile range.
\end{table}

\subsection*{{\bf B1}. Choosing the starting values}

Selecting the starting points is fundamental, because not all values of the parameters 
correspond to a well-defined quantile function. If $\xx_{it}^\T\betavec(\cdot \mid \hat\thetavec^{(0)})$
or $\zz_{i}^\T\gammavec(\cdot \mid \hat\phivec^{(0)})$ are not monotonically increasing functions,
the algorithm may converge to a nonsense solution.

As shown by equations (13--15), 
the gradient only depends on the model parameters through 
$u_{it}(\thetavec,\alpha_i)$ and $v_i(\phivec,\alpha_i)$, that correspond to the values of the cumulative distribution functions of $y_{it} - \alpha_i$ and $\alpha_i$,
respectively (equations 10 and 11). 
Given an initial estimate of $(u_{it},v_i)$, $(\hat u_{it}^{(0)}, \hat v_i^{(0)})$,
a starting value for $\thetavec$ and $\phivec$ can be obtained by approximating model (4) 
by a linear regression model of the form $$Y_{it} = \xx_{it}^\T\thetavec \bvec(\hat u_{it}^{(0)}) + \zz_{i}^\T\phivec\cvec(\hat v_i^{(0)}) + \epsilon_{it},$$
in which $\thetavec$ and $\phivec$ represent regression coefficients associated with the tensor products 
$\xx_i\otimes\bvec(\hat u_{it}^{(0)})$ and $\zz_i\otimes\cvec(\hat v_i^{(0)})$.

To obtain the initial values $\hat u_{it}^{(0)}$ and $\hat v_i^{(0)}$, we proceed as follows: 
first, we compute a preliminary estimate $\aahat^{(0)} = \{\hat \alpha_1^{(0)}, \ldots, \hat \alpha_N^{(0)}\}$
of $\aa = \{\alpha_1, \ldots, \alpha_N\}$ using the cluster medians.
Then, we use the \texttt{pch} R package to estimate nonparametrically the cumulative distribution function
of $y_{it} - \hat \alpha_i^{(0)}$ and, separately, that of $\aahat^{(0)}$. The fitted values are used to
define $\hat u_{it}^{(0)}$ and $\hat v_i^{(0)}$.

The algorithm appears very stable
and, in the simulation and data analysis conducted in this paper, never failed to converge.

\subsection*{{\bf B2}. Evaluating $\hat u_{it} := u_{it}(\hat \thetavec,\hat \alpha_i)$ and $\hat v_i := v_{i}(\hat \phivec,\hat \alpha_i)$}

At any current estimate $\hat\csi = (\hat\thetavec, \hat\phivec, \aahat)$,
evaluating the loss function and its derivatives requires computing the values 
$(\hat u_{it}, \hat v_i)$ such that 
$$\xx_{it}^\T\betavec(\hat u_{it} \mid \hat\thetavec) = \xx_{it}^\T\hat\thetavec\bvec(\hat u_{it}) = y_{it} - \hat\alpha_i,$$
$$\zz_i^\T\gammavec(\hat v_i \mid \hat\phivec) = \zz_i^\T\hat\phivec\cvec(\hat v_i) = \hat\alpha_i.$$
These values are not generally available in closed form, and are computed using a bisection algorithm.
For example, to compute $\hat u_{it}$, we proceed as follows: (i) start with $\hat u_{it}^{(0)} = 0.5$; (ii)
for $s = 1,2,3, \ldots$, define 
$$\hat u_{it}^{(s)} = \hat u_{it}^{(s - 1)} + \frac{1}{2^{s + 1}}\text{sign}\left(y_{it} - \hat\alpha_i - \xx_{it}^\T\hat\thetavec\bvec(\hat u_{it}^{(s - 1)})\right).$$
By bisecting the unit interval $20$ times, it is possible to achieve a precision of about $10^{-6}$. 

To solve equation (15), 
which includes both $u_{it}$ and $v_i$, two nested bisections must be implemented.
We start with $\hat v_i = 0.5$, and compute $\hat \alpha_i = \zz_i \hat\phi \cvec(\hat v_i)$, $i = 1, \ldots, N$.
We then apply bisection to compute the corresponding value of $\hat u_{it}$. Based on the sign of $G_{\alpha_i}(\hat \alpha_i, \hat\thetavec, \hat\phivec)$,
we update the current value of $\hat v_i$, and repeat the process until convergence.


\newpage
\section*{Appendix C - Extended simulation results}

We present additional simulation results that were not shown in the main text. We refer to Section 7 
for details on the simulation scenarios.

\subsection*{{\bf C1}. Comparison with standard penalized fixed-effects estimators}

We compared the performance of the described estimator with that of Koenker's (2004) penalized fixed-effects 
method for longitudinal quantile regression, implemented in the \texttt{rqpd} R package. 
To fit the model, the individual effects were centered around the median of their marginal distribution (which is zero
in simulation 1, and approximately $0.8744$ in simulation 2).
The tuning parameter $\lambda$ was selected as the ratio of the level-1 and level-2 variance components of 
a linear random-intercept model. Quantiles $0.2, 0.4, 0.6, 0.8$ were estimated jointly with equal weight.

A comparison is only possible for level-1 parameters, as no level-2 coefficients are estimated in
standard penalized quantile regression. Results are summarized in Tables \ref{C1a} and \ref{C1b}.
As expected, using a parametric model improved efficiency. Additionally, our estimator appeared
to have a smaller bias, as if the imposed parametric structure could alleviate the incidental parameter problem.


\subsection*{{\bf C2}. Model selection}
To assess the performance of the information criteria described in Section 6.2, 
we estimated three alternative models, of which only one was correctly specified.
Then, we used \textsc{AIC} and \textsc{BIC} to select the ``best'' model.
The simulation is described in details in Table \ref{C2a}. In simulation 1, we
compared different specifications of $\beta_1(u \mid \thetavec)$, using level-1 \textsc{AIC} and \textsc{BIC}.
In simulation 2, we compared different specifications of $\gamma_1(v \mid \phivec)$, using level-2 \textsc{AIC} and \textsc{BIC}.
In both scenarios, model II had the same number of parameters as model I, and corresponded to a very good approximation of the true model;
model III had less parameters, but was more severely misspecified. 

Results are summarized in Table \ref{C2b}. As expected, \textsc{BIC} tends to reward more parsimonious models,
especially when the sample size is relatively small.


\begin{table}
\footnotesize
\renewcommand\thetable{C1a}
\caption{Comparison with a penalized fixed-effects estimator ($N = 150$)}
\label{C1a}
\centering
\begin{tabular}{cc
m{0.01cm}m{0.5cm}m{0.5cm}m{0.5cm}m{0.3cm}m{0.3cm}
m{0.01cm}m{0.5cm}m{0.5cm}m{0.5cm}m{0.3cm}m{0.3cm}
}
\noalign{\vspace{0.2cm}}
\hline
\hline
\noalign{\vspace{0.2cm}}
 &&& \multicolumn{11}{c}{Simulation 1}\\
\noalign{\vspace{0.1cm}}
\hline
\noalign{\vspace{0.2cm}}
$T = 5$ & $u$ && 
	$\beta_0$ & $\hat\beta_0$ & $\hat\beta_0^K$ & $\text{se}$ & $\text{se}^K$ &&
	$\beta_1$ & $\hat\beta_1$ & $\hat\beta_1^K$ & $\text{se}$ & $\text{se}^K$\\
\cline{4-8}\cline{10-14}
\noalign{\vspace{0.1cm}}
& $0.2$ && 1.11 & 1.17 & 1.19 & .10 & .13 && 0.73 & 0.72 & 0.69 & .10 & .13 \\
& $0.4$ && 1.26 & 1.30 & 1.33 & .11 & .12 && 0.99 & 0.97 & 0.92 & .11 & .11 \\
& $0.6$ && 1.46 & 1.48 & 1.48 & .11 & .13 && 1.01 & 0.99 & 1.02 & .11 & .12 \\
& $0.8$ && 1.80 & 1.79 & 1.75 & .12 & .15 && 1.27 & 1.24 & 1.23 & .12 & .22 \\
\noalign{\vspace{0.2cm}}
$T = 10$ & $u$ && 
	$\beta_0$ & $\hat\beta_0$ & $\hat\beta_0^K$ & $\text{se}$ & $\text{se}^K$ &&
	$\beta_1$ & $\hat\beta_1$ & $\hat\beta_1^K$ & $\text{se}$ & $\text{se}^K$\\
\cline{4-8}\cline{10-14}
\noalign{\vspace{0.1cm}}
& $0.2$ && 1.11 & 1.14 & 1.16 & .10 & .12 && 0.73 & 0.73 & 0.71 & .06 & .09 \\
& $0.4$ && 1.26 & 1.27 & 1.30 & .10 & .12 && 0.99 & 0.98 & 0.95 & .07 & .07 \\
& $0.6$ && 1.46 & 1.47 & 1.48 & .10 & .12 && 1.01 & 1.00 & 1.03 & .07 & .08 \\
& $0.8$ && 1.80 & 1.80 & 1.78 & .10 & .14 && 1.27 & 1.25 & 1.25 & .08 & .17 \\
\noalign{\vspace{0.2cm}}
\hline
\hline
\noalign{\vspace{0.2cm}}
 &&& \multicolumn{11}{c}{Simulation 2}\\
\noalign{\vspace{0.1cm}}
\hline
\noalign{\vspace{0.2cm}}
$T = 5$ & $u$ && 
	$\beta_0$ & $\hat\beta_0$ & $\hat\beta_0^K$ & $\text{se}$ & $\text{se}^K$ &&
	$\beta_1$ & $\hat\beta_1$ & $\hat\beta_1^K$ & $\text{se}$ & $\text{se}^K$\\
\cline{4-8}\cline{10-14}
\noalign{\vspace{0.1cm}}
& $0.2$ && 0.11 & 0.11 & 0.22 & .01 & .11 && 3.60 & 3.79 & 3.71 & .11 & .17 \\
& $0.4$ && 0.24 & 0.24 & 0.35 & .03 & .11 && 4.20 & 4.28 & 4.17 & .12 & .18 \\
& $0.6$ && 0.41 & 0.41 & 0.46 & .05 & .11 && 4.80 & 4.77 & 4.72 & .15 & .19 \\
& $0.8$ && 0.66 & 0.66 & 0.66 & .08 & .12 && 5.40 & 5.26 & 5.29 & .19 & .21 \\
\noalign{\vspace{0.2cm}}
$T = 10$ & $u$ && 
	$\beta_0$ & $\hat\beta_0$ & $\hat\beta_0^K$ & $\text{se}$ & $\text{se}^K$ &&
	$\beta_1$ & $\hat\beta_1$ & $\hat\beta_1^K$ & $\text{se}$ & $\text{se}^K$\\
\cline{4-8}\cline{10-14}
\noalign{\vspace{0.1cm}}
& $0.2$ && 0.11 & 0.11 & 0.19 & .01 & .09 && 3.60 & 3.68 & 3.64 & .07 & .10 \\
& $0.4$ && 0.24 & 0.23 & 0.31 & .02 & .09 && 4.20 & 4.23 & 4.18 & .08 & .12 \\
& $0.6$ && 0.41 & 0.40 & 0.44 & .03 & .10 && 4.80 & 4.78 & 4.76 & .10 & .13 \\
& $0.8$ && 0.66 & 0.64 & 0.66 & .05 & .10 && 5.40 & 5.33 & 5.35 & .12 & .15 \\
\noalign{\vspace{0.1cm}}
\hline
\hline
\noalign{\vspace{0.1cm}}
\end{tabular}
\vspace{0.1cm}\footnotesize\center
Comparison with Koenker's (2004) penalized fixed-effects estimator with $N = 150$ and $T = \{5, 10\}$.
We report the true value ($\beta$) of the level-1 coefficients at the quintiles ($0.2, 0.4, 0.6, 0.8$),
and the average estimates and empirical standard errors obtained with our method ($\hat \beta$, $\text{se}$) and 
with Koenker's penalized fixed-effects estimator ($\hat \beta^K$, $\text{se}^K)$, across $B = 1000$ simulated datasets. 
\end{table}


\begin{table}
\footnotesize
\renewcommand\thetable{C1b}
\caption{Comparison with a penalized fixed-effects estimator ($N = 300$)}
\label{C1b}
\centering
\begin{tabular}{cc
m{0.01cm}m{0.5cm}m{0.5cm}m{0.5cm}m{0.3cm}m{0.3cm}
m{0.01cm}m{0.5cm}m{0.5cm}m{0.5cm}m{0.3cm}m{0.3cm}
}
\noalign{\vspace{0.2cm}}
\hline
\hline
\noalign{\vspace{0.2cm}}
 &&& \multicolumn{11}{c}{Simulation 1}\\
\noalign{\vspace{0.1cm}}
\hline
\noalign{\vspace{0.2cm}}
$T = 5$ & $u$ && 
	$\beta_0$ & $\hat\beta_0$ & $\hat\beta_0^K$ & $\text{se}$ & $\text{se}^K$ &&
	$\beta_1$ & $\hat\beta_1$ & $\hat\beta_1^K$ & $\text{se}$ & $\text{se}^K$\\
\cline{4-8}\cline{10-14}
\noalign{\vspace{0.1cm}}
& $0.2$ && 1.11 & 1.17 & 1.20 & .08 & .09 && 0.73 & 0.72 & 0.69 & .07 & .09 \\
& $0.4$ && 1.26 & 1.30 & 1.33 & .08 & .09 && 0.99 & 0.97 & 0.92 & .08 & .08 \\
& $0.6$ && 1.46 & 1.48 & 1.48 & .08 & .09 && 1.01 & 1.00 & 1.02 & .08 & .09 \\
& $0.8$ && 1.80 & 1.79 & 1.74 & .08 & .11 && 1.27 & 1.24 & 1.24 & .09 & .16 \\
\noalign{\vspace{0.2cm}}
$T = 10$ & $u$ && 
	$\beta_0$ & $\hat\beta_0$ & $\hat\beta_0^K$ & $\text{se}$ & $\text{se}^K$ &&
	$\beta_1$ & $\hat\beta_1$ & $\hat\beta_1^K$ & $\text{se}$ & $\text{se}^K$\\
\cline{4-8}\cline{10-14}
\noalign{\vspace{0.1cm}}
& $0.2$ && 1.11 & 1.14 & 1.15 & .07 & .08 && 0.73 & 0.73 & 0.71 & .05 & .06 \\
& $0.4$ && 1.26 & 1.27 & 1.29 & .07 & .08 && 0.99 & 0.98 & 0.94 & .05 & .05 \\
& $0.6$ && 1.46 & 1.47 & 1.47 & .07 & .09 && 1.01 & 1.00 & 1.02 & .05 & .06 \\
& $0.8$ && 1.80 & 1.80 & 1.78 & .07 & .10 && 1.27 & 1.25 & 1.25 & .05 & .12 \\
\noalign{\vspace{0.2cm}}
\hline
\hline
\noalign{\vspace{0.2cm}}
 &&& \multicolumn{11}{c}{Simulation 2}\\
\noalign{\vspace{0.1cm}}
\hline
\noalign{\vspace{0.2cm}}
$T = 5$ & $u$ && 
	$\beta_0$ & $\hat\beta_0$ & $\hat\beta_0^K$ & $\text{se}$ & $\text{se}^K$ &&
	$\beta_1$ & $\hat\beta_1$ & $\hat\beta_1^K$ & $\text{se}$ & $\text{se}^K$\\
\cline{4-8}\cline{10-14}
\noalign{\vspace{0.1cm}}
& $0.2$ && 0.11 & 0.11 & 0.22 & .01 & .08 && 3.60 & 3.79 & 3.70 & .08 & .11 \\
& $0.4$ && 0.24 & 0.24 & 0.36 & .02 & .08 && 4.20 & 4.28 & 4.17 & .08 & .13 \\
& $0.6$ && 0.41 & 0.41 & 0.47 & .03 & .08 && 4.80 & 4.77 & 4.72 & .10 & .14 \\
& $0.8$ && 0.66 & 0.66 & 0.65 & .06 & .08 && 5.40 & 5.27 & 5.30 & .13 & .15 \\
\noalign{\vspace{0.2cm}}
$T = 10$ & $u$ && 
	$\beta_0$ & $\hat\beta_0$ & $\hat\beta_0^K$ & $\text{se}$ & $\text{se}^K$ &&
	$\beta_1$ & $\hat\beta_1$ & $\hat\beta_1^K$ & $\text{se}$ & $\text{se}^K$\\
\cline{4-8}\cline{10-14}
\noalign{\vspace{0.1cm}}
& $0.2$ && 0.11 & 0.10 & 0.18 & .01 & .06 && 3.60 & 3.67 & 3.64 & .05 & .07 \\
& $0.4$ && 0.24 & 0.23 & 0.31 & .01 & .07 && 4.20 & 4.23 & 4.17 & .06 & .08 \\
& $0.6$ && 0.41 & 0.40 & 0.44 & .02 & .07 && 4.80 & 4.78 & 4.76 & .07 & .09 \\
& $0.8$ && 0.66 & 0.64 & 0.65 & .04 & .07 && 5.40 & 5.33 & 5.35 & .09 & .10 \\
\noalign{\vspace{0.1cm}}
\hline
\hline
\noalign{\vspace{0.1cm}}
\end{tabular}
\vspace{0.1cm}\footnotesize\center
Comparison with Koenker's (2004) penalized fixed-effects estimator with $N = 300$ and $T = \{5, 10\}$.
\end{table}


\begin{table}
\footnotesize
\renewcommand\thetable{C2a}
\caption{Three different model specifications}
\label{C2a}
\centering
\begin{tabular}{ccc
}
\noalign{\vspace{0.2cm}}
\hline
\hline
\noalign{\vspace{0.2cm}}
 & Simulation 1 & Simulation 2\\
\noalign{\vspace{0.1cm}}
\hline
\noalign{\vspace{0.2cm}}
Coefficient & $\beta_1(u \mid \thetavec)$ & $\gamma_1(v \mid \phivec)$\\
True value & $1 + 10(u - 0.5)^3$ & $0.5\log(1 - \log(1 - v))$\\ 
Model I (correct) & $\theta_{10} + \theta_{11}u + \theta_{12}u^2 + \theta_{13}u^3$ & $\phi_{10} + \phi_{11}\log(1 - \log(1 - v))$\\
Model II (misspecified) & $\theta_{10} + \theta_{11}u + \theta_{12}\cos{(\pi u)} + \theta_{13}\sin{(\pi u)}$ & $\phi_{10} + \phi_{11}(1 - (1 - v)^{0.5})$\\
Model III (misspecified) & $\theta_{10} + \theta_{11}\sqrt{u} + \theta_{12}\sqrt{1 - u}$ &  $\phi_{11}(1 - v^2)$\\
\noalign{\vspace{0.1cm}}
\hline
\hline
\end{tabular}
\vspace{0.1cm}\footnotesize\center
Alternative models to be compared using \textsc{AIC} and \textsc{BIC}. For each scenario,
we estimated three different models, of which one (model I) was correctly specified, and the other two (models II and III) were misspecified.
\end{table}

\begin{table}
\footnotesize
\renewcommand\thetable{C2b}
\caption{Performance of \textsc{AIC} and \textsc{BIC}}
\label{C2b}
\centering
\begin{tabular}{ccccccccccccc}
\noalign{\vspace{0.2cm}}
\hline
\hline
\noalign{\vspace{0.2cm}}
 && \multicolumn{11}{c}{Simulation 1}\\
\noalign{\vspace{0.1cm}}
\hline
\noalign{\vspace{0.2cm}}
&& \multicolumn{5}{c}{$N = 150$} && \multicolumn{5}{c}{$N = 300$}\\
\noalign{\vspace{0.1cm}}
\cline{3-7}\cline{9-13}
\noalign{\vspace{0.1cm}}
&& \multicolumn{2}{c}{$T = 5$} && \multicolumn{2}{c}{$T = 10$} && \multicolumn{2}{c}{$T = 5$} && \multicolumn{2}{c}{$T = 10$}\\
\cline{3-4}\cline{6-7}\cline{9-10}\cline{12-13}
\noalign{\vspace{0.1cm}}
&& \textsc{AIC} & \textsc{BIC} && \textsc{AIC} & \textsc{BIC} && \textsc{AIC} & \textsc{BIC} && \textsc{AIC} & \textsc{BIC}\\
\noalign{\vspace{0.1cm}}
Model I (correct) && 				0.77 & 0.28 && 0.88 & 0.84 && 0.87 & 0.82 && 0.94 & 0.94\\
Model II (misspecified) &&			0.21 & 0.07 && 0.12 & 0.12 && 0.13 & 0.12 && 0.06 & 0.06\\
Model III (misspecified) &&		0.02 & 0.65 && 0.00 & 0.04 && 0.00 & 0.06 && 0.00 & 0.00\\
\noalign{\vspace{0.2cm}}
\hline
\hline
\noalign{\vspace{0.2cm}}
 && \multicolumn{11}{c}{Simulation 2}\\
\noalign{\vspace{0.1cm}}
\hline
\noalign{\vspace{0.2cm}}
&& \multicolumn{5}{c}{$N = 150$} && \multicolumn{5}{c}{$N = 300$}\\
\noalign{\vspace{0.1cm}}
\cline{3-7}\cline{9-13}
\noalign{\vspace{0.1cm}}
&& \multicolumn{2}{c}{$T = 5$} && \multicolumn{2}{c}{$T = 10$} && \multicolumn{2}{c}{$T = 5$} && \multicolumn{2}{c}{$T = 10$}\\
\cline{3-4}\cline{6-7}\cline{9-10}\cline{12-13}
\noalign{\vspace{0.1cm}}
&& \textsc{AIC} & \textsc{BIC} && \textsc{AIC} & \textsc{BIC} && \textsc{AIC} & \textsc{BIC} && \textsc{AIC} & \textsc{BIC}\\
\noalign{\vspace{0.1cm}}
Model I (correct) && 				0.71 & 0.55 && 0.75 & 0.64 && 0.86 & 0.81 && 0.80 & 0.79\\
Model II (misspecified) &&			0.20 & 0.15 && 0.23 & 0.20 && 0.13 & 0.12 && 0.20 & 0.20\\
Model III (misspecified) &&		0.09 & 0.30 && 0.02 & 0.16 && 0.01 & 0.07 && 0.00 & 0.01\\
\noalign{\vspace{0.1cm}}
\hline
\hline
\end{tabular}
\vspace{0.1cm}\footnotesize\center
Model selection based on \textsc{AIC} and \textsc{BIC}. In the table, we report the relative frequency with which
each of the three candidate models was selected, across $B = 1000$ simulated datasets.
The correct specification is that defined by model I (see Table \ref{C2a}). 
\end{table}

\end{document}